\documentclass[draftcls, onecolumn, 11pt]{IEEEtran}
\usepackage{amsmath}
\usepackage{amssymb}
\usepackage{amsfonts}
\usepackage{graphicx}
\usepackage{epsfig}
\usepackage{subfigure}
\usepackage{psfrag}
\usepackage{xcolor}
\usepackage{amsfonts, bm}
\usepackage{theorem}

\newtheorem{lemma}{Lemma}
\newtheorem{prop}{Proposition}
\newtheorem{theorem}{Theorem}

\newtheorem{corollary}{Corollary}
\newtheorem{remark}{Remark}

\theorembodyfont{\rmfamily}
\newenvironment{proof}[1][Proof]{\begin{trivlist}
\item[\hskip \labelsep {\bfseries #1}]}{\end{trivlist}}

\newcommand{\st}{{\rm s.t.}}

\newcommand{\bC}{{\bf C}}

\newcommand{\bv}{{\bm{v}}}

\newcommand{\bu}{\bm{u}}
\newcommand{\baru}{{\bar u}}

\newcommand{\bA}{{\textbf{A}}}

\newcommand{\bx}{\mathbf{x}}

\newcommand{\bH}{\mathbf{H}}
\newcommand{\bB}{\mathbf{B}}

\newcommand{\bV}{\mathbf{V}}
\newcommand{\bI}{\mathbf{I}}
\newcommand{\bD}{\mathbf{D}}

\newcommand{\bOmega}{\mathbf{\Omega}}
\newcommand{\bPsi}{\mathbf{\Psi}}
\newcommand{\bUpsilon}{\mathbf{\Upsilon}}

\newcommand{\sinr}{\textrm{SINR}}
\newcommand{\trace}{\textrm{Tr}}
\begin{document}

\title{SINR Constrained Beamforming for a MIMO Multi-user Downlink System}
%A SDP Approach for Range-free Localization in Wireless Sensor Network

\author{
\authorblockN{Qingjiang Shi, Meisam Razaviyayn, Mingyi Hong and Zhi-Quan Luo}
%\authorblockA{$^1$Department of Communication Engineering, Zhejiang Sci-Tech University, Hangzhou, China}
%\authorblockA{$^2$Department of Electrical and Computer Engineering, University of Minnesota, Minneapolis, MN, USA}
\thanks{Q. Shi is with the School of Information and Science Technology, Zhejiang Sci-Tech University, Hangzhou 310018, China. Email: qing.j.shi@gmail.com.}
%\thanks{M. Razaviyayn is with the Dept. of Electrical and Computer Engineering, Stanford University, USA. Email: meisamr@stanford.edu.}
%\thanks {M. Hong is with the Dept. of Industrial and Manufacturing Systems Engineering, Iowa State University, IA 50011, USA. Email: mingyi@iastate.edu}
%\thanks{Z.-Q. Luo is with the Dept. of Electrical and Computer Engineering, University of Minnesota, Minneapolis, MN, USA. Email: luozq@ece.umn.edu}
\thanks{M. Razaviyayn, M. Hong, and Z.-Q. Luo were all with the Dept. of Electrical and Computer Engineering, University of Minnesota, Minneapolis, MN, USA. Email: luozq@ece.umn.edu.}
}

\maketitle

\begin{abstract}
Consider a multi-input multi-output (MIMO) downlink multi-user channel. A well-studied problem in such system is the design of linear beamformers for power minimization with the quality of service (QoS) constraints. The most representative algorithms for solving this class of problems are the so-called MMSE-SOCP algorithm \cite{Visotsky1999,Wong2005} and the UDD algorithm \cite{Codreanu2007}. The former is based on alternating optimization of the transmit and receive beamformers; while the latter is based on the well-known uplink-dowlink duality theory. Despite their wide applicability, the convergence (to KKT solutions) of both algorithms is still open in the literature. In this paper, we rigorously establish the convergence of these algorithms for QoS-constrained power minimization (QCPM) problem with both single stream and multiple streams per user cases. Key to our analysis is the development and analysis of a new MMSE-DUAL algorithm, which connects the MMSE-SOCP and the UDD algorithm.
%first investigate the QoS-constrained power minimization (QCPM) problem of single stream case with a thorough KKT analysis. Based on the KKT analysis result,
%we then propose a low complexity iterative algorithm called MMSE-DUAL algorithm, which is equivalent to the MMSE-SOCP algorithm but plays a key role in establishing the connections between the MMSE-SOCP algorithm and the UDD algorithm.
%Our theoretical analysis shows that the MMSE-SOCP and the UDD algorithm can both monotonically converge to the set of KKT points of the QCPM problem. Finally we extends these algorithms to the multiple stream case and prove that they can also reach a stationary point of the QCPM problem of multiple stream case under a mild condition.
Our numerical experiments show that 1) all these algorithms can almost always reach points with the same objective value irrespective of initialization; 2) the MMSE-SOCP/MMSE-DUAL algorithm works well while the UDD algorithm may fail with an infeasible initialization.
\end{abstract}
%{\color{red}[[I have reorganized the paper and rewritten many part of the paper.]]}
\maketitle

\section{Introduction}
Multi-user MIMO (MU-MIMO) is a key building block of the next generation wireless communication system. In a MU-MIMO downlink system, a base station (BS) equipped with multiple antennas
simultaneously transmits data to a group of multiple antenna users. The multi-user interference, which is the major performance limiting factor of MU-MIMO systems, must be managed intelligently using the physical layer techniques such as beamforming. In general, there are
two major objectives in the beamformer design problem. One is to
maximize a system utility (e.g., throughput) under some power
constraint, while the other one is to minimize the system transmission
power subject to QoS requirements \cite{hong12survey}. Although both formulations are justifiable and well-studied, the latter is more appropriate for scenarios where the users need a guaranteed QoS level.

Although the QoS-constrained power minimization (QCPM) problem is globally solvable in polynomial time for MISO system \cite{Ottersten2001,Yuwei2010,Wiesel06}, this problem is highly non-convex and difficult to solve in MIMO systems.  In fact, it has been shown that when the BSs and the users are equipped with more than two antennas, the problem becomes NP-hard \cite{Razaviyayn12maxmin}. Therefore, many algorithms have been proposed to solve this problem approximately.
For example, the references \cite{Spencer2002} and \cite{Pan2004} propose algorithms based on interference nulling, which can completely eliminate the inter-user interference. In these algorithms, the search space of transmit and
receive beamformers is limited to zero forcing transceivers, leading
to simple but suboptimal solutions. In addition, such interference nulling based methods require that the number of transmit antennas is no less than the total number of active users, which is
impractical in many scenarios.

Another approach for solving the MU-MIMO downlink QCPM problem is based on the iterative
optimization methods \cite{Khachan2006,Codreanu2007, Chang2002,Wong2005,Visotsky1999}.
References \cite{Khachan2006,Codreanu2007, Chang2002} provide
iterative algorithms that update the transmit beamformers, receive beamformers,
and transmit powers\footnote{In \cite{Khachan2006,Codreanu2007, Chang2002},
the power allocation and beamforming are separated and thus both
need to be optimized.} by switching between the downlink and the
dual uplink channels. Central to these methods is the notion of
uplink-downlink duality (UDD) theory \cite{Boche2001,Viswanath2003,Hunger2008}, which guarantees that
{a set of target SINR levels is achievable in the downlink channel
if and only if the same set of SINR levels is achievable in the corresponding dual
uplink channel}. We refer to such algorithms as UDD algorithms. The UDD algorithm was
first proposed in \cite{Chang2002} for the multi-antenna case, where the \emph{receive} (resp. transmit) beamformer update is followed immediately by the \emph{transmit} (resp. receive) power update. The UDD algorithm of \cite{Codreanu2007} differs from that of \cite{Chang2002} in the order of updating transmit/receive beamformers and powers. In the UDD algorithm of \cite{Codreanu2007}, the \emph{transmit} (resp. receive) power is updated exactly after the \emph{transmit} (resp. receive) beamformer. Importantly, it is shown in \cite{Codreanu2007} that the UDD algorithm monotonically decreases the total power consumption while satisfying the QoS constraints. However, the algorithms in \cite{Chang2002,Codreanu2007} can only apply to the \emph{single stream case}. In \cite{Khachan2006} the UDD algorithm has been generalized to the multiple stream case under the assumption of no joint detection at the receivers (i.e., \emph{inter-stream} interference is considered). Notice that, to the best of our knowledge, it is still not known whether the UDD algorithm converges to a KKT solution.

Different from the previous works, the references \cite{Visotsky1999} and \cite{Wong2005} have proposed an iterative algorithm named MMSE-SOCP, which aims to solve the QCPM problem directly. The algorithm consists of the following two key steps: 1) Fixing the transmit power, update the receive beamformers using MMSE receiver,
and 2) Fixing the receive beamformers, update the transmit beamformers by solving the power minimization problem with respect
to the transmit beamformers. The authors show that the total transmit power monotonically decreases and thus converges. However, the convergence of the MMSE-SOCP iterates  to a KKT solution is not known. Nevertheless, the authors of \cite{Chang2002} and \cite{Wong2005} observe that the MMSE-SOCP algorithm always generates a sequence that converges to a unique solution, irrespective of the
initial point (or at least with high probability, see
\cite{Chang2002}). Hence, a conjecture has been made in
\cite{Wong2005}, stating that the MMSE-SOCP algorithm probably reaches a local optimum solution of the QCPM problem\footnote{It is argued in \cite{Wong2005} that, ``Though the proposed algorithm always converges and seems
to converge to a unique optima irrespective of the starting point
from the simulation results,it may be possible that the steadystate
solution is a \emph{local optimum.}''}

In this paper, we settle the convergence issue related to the MMSE-SOCP and the UDD algorithm. We show that both algorithms converge globally to the set of KKT solutions, regardless of the number of streams intended for each user. We start by analyzing the KKT conditions of the QCPM problem of the single stream case. Based on the analysis, we propose a novel iterative algorithm called MMSE-DUAL algorithm, which is
essentially equivalent to the MMSE-SOCP algorithm, but with the added benefits of having almost closed form updates and a lower complexity. Through the MMSE-DUAL algorithm,
we reveal some connections between the UDD algorithm\cite{Codreanu2007} and the MMSE-SOCP algorithm \cite{Visotsky1999,Wong2005}.
More importantly we prove that both the MMSE-DUAL algorithm and the UDD algorithm monotonically converge to the set of  KKT solutions of the QCPM problem. In addition, we extend the algorithms to the multiple stream case and prove that they can also reach a KKT point of the QCPM problem under some mild conditions. As will be seen later, the MMSE-DUAL algorithm has a lower complexity than the UDD algorithm if the number of streams is greater than the number of transmit antennas. Moreover, although both algorithms require feasible initialization, it is easier for the MMSE-DUAL algorithm to obtain a feasible initialization (see Remark \ref{rem:feasibility} in Section IV).

The remainder of this paper is organized as follows. In
Section II we give the formulation of the QCPM problem and provide a brief review of the existing algorithms. In Section III we propose the MMSE-DUAL algorithm which reveals the relation between the  MMSE-SOCP and the UDD algorithm. Then, we state the convergence results of the three algorithms in Section IV and extend the algorithms to the multiple stream case in Section V. Finally, section VI presents some simulation results and Section VII concludes the paper.

\emph{Notations}: Throughout this paper, we use uppercase bold letters for matrices,
lowercase bold letters for column vectors, and regular letters for scalars. The superscript $^H$ is used to denote the Hermitian transpose of a matrix. For a complex number $a$, $\angle(a)$ and $\textrm{Im}(a)$ denote its phase angle and imaginary part, respectively. For a function $f(\cdot)$, $\nabla_{\bx} f(\cdot)$ denotes its gradient with respect to the variable $\bx$. For a matrix $\bA$, $\bA\succeq0$ indicates that $\bA$ is a positive semidefinite matrix. $\bI$ denotes the identity matrix of an appropriate size. The circularly symmetric complex Gaussian distribution is represented by $\mathcal{CN}(\mu, \sigma^2)$, where $\mu$ is the mean and $\sigma^2$ is the variance of the distribution. The notations $\trace(\cdot)$ and $\det(\cdot)$ represent the trace and the determinant operator, respectively.

\section{Problem Formulation And Existing Algorithms}
\subsection{Problem Formulation}
Consider a multi-user MIMO downlink system with $K$ users, where the
BS is equipped with $M > 1$ antennas and each user~$k$ is equipped
with $N_k > 1$ antennas. Let us use $\mathcal{K} = \{1,2,\ldots,K\}$  to denote the set
of all users. Assume for now that
the BS transmits the single stream $s_k$ to the intended
receiver~$k$ with no multiplexing (the multiple stream case will be considered in Section VI). Let us also assume that the BS utilizes the transmit
beamformer~$\bv_k\in\mathbb{C}^{M\times 1}$  to send the data stream $s_k$ to user $k$. On the other side, user $k$ utilizes the receive beamformer~$\mathbf{u}_k \in \mathbb{C}^{N_k \times 1}$ to estimate its transmitted data stream. The estimated data stream $\hat{s}_k$ can be mathematically expressed as%\footnote{We remark that, although our MU-MIMO system model only considers single stream, the techniques developed in this paper apply to the multistream case with stream-wise QoS constraint.}
\begin{equation}\label{eq:CH_model}
\hat{s}_k=\bu_k^H\left(\mathbf{H}_k\sum_{j=1}^K \bv_js_j+\bm{n}_k\right), \;\; \forall k\in \mathcal{K},
\end{equation}
where $\mathbf{H}_k\in\mathbb{C}^{N_k\times M}$ denotes the channel
matrix from the BS to the receiver~$k$;
$\bm{n}_k\in\mathbb{C}^{N_k\times 1}$ is the additive white Gaussian
noise (AWGN) with distribution $\mathcal{CN}(0, \sigma_k^2)$. The data
streams $s_k$'s are i.i.d. and independent of the noise level; and
have distribution $\mathcal{CN}(0,1)$.

We are interested in designing the transmit and receive beamformers to minimize the transmit power while the users' QoS requirements are satisfied.  Let us consider the signal-to-interference-plus-noise ratio (SINR) as the QoS measure. The SINR of user $k$ is given by:
\begin{align} \label{eq:SINRk}
\sinr_k\triangleq\frac{|\bu_k^H\mathbf{H}_k\bv_k|^2}{\sum_{j\neq k}|\bu_k^H\mathbf{H}_k\bv_j|^2+\sigma_k^2 \|{\bu_k}\|^2}.
\end{align}

Mathematically, the QCPM problem can be written as
\begin{equation*}
\label{eq:P1}
\textrm{(P1)}
\begin{split}
\min_{\bv,\bu} \; &\sum_{k=1}^K \|\bv_k\|^2\\
\textrm{s.t.}\;\; &\frac{|\bu_k^H\mathbf{H}_k\bv_k|^2}{\sum_{j\neq k}|\bu_k^H\mathbf{H}_k\bv_j|^2{+}\sigma_k^2\|\bu_k\|^2}{\geq} \gamma_k, \;\forall k\in \mathcal{K},
\end{split}
\end{equation*}
where $\gamma_k > 0$ is the intended SINR level of user~$k$;  the
set of all transmit beamformers (resp. receive beamformers) is
denoted by $\bv \triangleq \{\bv_1,\ldots, \bv_K\}$ (resp.
$\bu\triangleq \{\bu_1,\ldots, \bu_K\}$). Throughout the paper, we
assume that problem (P1) is feasible and $\sigma_k^2>0$ for all $k$.

%It is readily seen that problem (P1) is nonconvex and hard to solve.
%In this paper, we propose a simple iterative algorithm which seems like alternating optimization algorithm. However, due to the coupled constraints on $\bv_k$s and $\bu_k$s, the theory developed in \cite[Proposition 2.7.1]{Bertsekas_book} does not apply here. Observing the special structure of the problem (P1), we prove that the proposed algorithm, the MMSE-SOCP, and the UDD algorithm all can monotonically converge to a KKT point of the QCPM problem.
\subsection{Existing Algorithms: A Brief Review}
In this subsection, we briefly review the existing MMSE-SOCP algorithm
\cite{Wong2005} and the UDD  algorithm \cite{Codreanu2007}.
\subsubsection{MMSE-SOCP Algorithm \cite{Visotsky1999,Wong2005}}
The MMSE-SOCP algorithm alternates between the following two steps:
\begin{itemize}
\item[1.] Fixing all the transmit beamformers, update the receive beamformers using the MMSE receiver, i.e.,
$$\bu_k=\left(\sum_{j\neq k}\mathbf{H}_k\bv_j\bv_j^H\mathbf{H}_k^H +\sigma_k^2\mathbf{I}\right)^{-1}\mathbf{H}_k\bv_k,\;\forall k \in \mathcal{K}.$$
\item[2.]  Fixing all the receive beamformers, update the transmit beamformers by solving
\begin{equation}
\label{eq:P1-socp}
\begin{split}
\min_{\bv} &\sum_{k=1}^K \|\bv_k\|^2\\
\textrm{s.t.}\;\; &\frac{|\bu_k^H\mathbf{H}_k\bv_k|^2}{\sum_{j\neq k}|\bu_k^H\mathbf{H}_k\bv_j|^2{+}\sigma_k^2\|\bu_k\|^2}{\geq} \gamma_k, \;\forall k\in \mathcal{K},
\end{split}
\end{equation}
 which can be transformed to a second-order cone program (SOCP) \cite{Wong2005,Ottersten2001}.
 %the following second-order cone program (SOCP) through phase rotation\cite{Wong2005,Ottersten2001}
% \begin{equation}
%\label{eq:P1-socp-formulation}
%\begin{split}
%\min_{\bv} &\sum_{k=1}^K \|\bv_k\|^2\\
%\textrm{s.t.}\;\; &\frac{1}{\sqrt{\gamma_k}}\bu_k^H\mathbf{H}_k\bv_k{\geq} \sqrt{\sum_{j\neq k}|\bu_k^H\mathbf{H}_k\bv_j|^2{+}\sigma_k^2\|\bu_k\|^2}, \;\forall k\in \mathcal{K}.
%\end{split}
%\end{equation}
 The above SOCP has $KM$ unknowns and can be efficiently solved via interior-point algorithm; each iteration of the interior-point algorithm has computational complexity of $\mathcal{O}(K^3M^3)$\cite{cvx_book}. %Reference %\cite{Wong2005} has shown that the MMSE-SOCP algorithm can decrease the transmit power consumption as the %iteration proceeds.
\end{itemize}

%
%\begin{table}[h]
%\centering
%\caption{The MMSE-SOCP algorithm}\label{fig:pseudo_code_UDD}
%\begin{tabular}{|p{3.3in}|}
%\hline
%\begin{itemize}
%\item \textbf{Input}: $\mathbf{H}_k$, $\sigma_k^2$, $\gamma_k$, $k=1,2,\ldots,K$%\diag[g_{k1}, g_{k2}, \ldots, g_{kL}], \forall k$,  are the frequency--domain channel responses, $\epsilon$ is a predefined small scalar, $T$ is the maximum total number of iterations.
%\item \textbf{Output}: the transmit beamformer $\bv_k$ and $\bu_k$, $k=1,2,\ldots,K$.
%\item [1] \; set $t=0$ and $N\geq1$
%\item[2] \; initialize $\bv_k$ and $\lambda_k$, $k=1,2,\ldots,K$.
%\item [3] \; \textbf{repeat}
%\item [4]\; \quad $t \longleftarrow t+1$;
%\item [5] \quad  $\bm{u}_k\longleftarrow\left(\sum_{j\neq k}\mathbf{H}_k\bv_j\bv_j^H\mathbf{H}_k^H +\sigma_k^2\mathbf{I}\right)^{-1}\mathbf{H}_k\bv_k$, $\forall k$
%\item [6] \quad update $\bv_k$s by solving the SOCP problem \eqref{eq:P1-socp}
%%\item [] \quad\quad\quad$\min_{\bv_ks} \sum_{k=1}^K \|\bv_k\|^2$
%%\item [] \quad\quad\quad s.t. $\frac{|\bu_k^H\mathbf{H}_k\bv_k|^2}{\sum_{j\neq k}|\bu_k^H\mathbf{H}_k\bv_j|^2{+}\sigma_k^2\|\bu_k\|^2}{\geq} \gamma_k, k=1,2,\ldots,K.$
%\item [7]\; \textbf{until} meet some convergence criterion
%%\item [10]\; $\bP_k=\bF\bX_k\bG_k^H, \forall k$
%\end{itemize}
%\\
%\hline
%\end{tabular}
%\end{table}

\subsubsection{UDD Algorithm \cite{Codreanu2007}}
 Let $\bar\bu_k$ and $\bar\bv_k$ denote the normalized beamformer, i.e.,  $\bu_k=\sqrt{q_k}\bar{\bu}_k$ and $\bv_k=\sqrt{p_k}\bar{\bv}_k$ with $\|\bar{\bu}_k\|= \|\bar{\bv}_k\| =1$. We refer $\bar{\bu}_k$ and $\bar{\bv}_k$ as the normalized beamformers, $p_k$ and $q_k$ as the power consumption. Using these notations, the downlink channel model \eqref{eq:CH_model} can be rewritten as
\begin{equation}\label{eq:dlink_chm}
\hat{s}_k=\sqrt{q_k}\bar{\bu}_k^H\left(\mathbf{H}_k\sum_{j=1}^K \sqrt{p_j}\bar{\bv}_js_j+\bm{n}_k\right),\;\forall k\in \mathcal{K}.
\end{equation}

The UDD algorithm is established by introducing a \emph{virtual} uplink channel, which can be constructed through the following three steps: 1) reverse the directions of all links; 2) replace the channel matrices by their conjugated transposed version (i.e., $\bH_k\leftarrow\bH_k^H$ $\forall k$); 3) exchange the role of transmit beamformer and receive beamformer. Mathematically, the virtual uplink channel model can be represented as
\begin{equation}\label{eq:ulink_chm}
\tilde{s}_k=p_k\bar{\bv}_k^H\left(\sum_{j=1}^K q_j \mathbf{H}_j^H\bar{\bu}_js_j+\tilde{\bm{n}}_k\right), \;\forall k\in \mathcal{K},
\end{equation}
where $\tilde{\bm{n}}_k=\frac{1}{\sigma_k}\bm{n}_k$ is the virtual uplink channel AWGN.

In terms of the channel models \eqref{eq:dlink_chm} and \eqref{eq:ulink_chm}, the SINRs for the downlink and uplink are respectively expressed as
\begin{equation}
\sinr_k^{\textrm{D}}=\frac{p_k\|\bar\bu_k^H\mathbf{H}_k\bar{\bv}_k\|^2}{\bar\bu_k^H\left(\sum_{j\neq k}p_j\mathbf{H}_k\bar{\bv}_j\bar{\bv}_j^H\mathbf{H}_k^H+\sigma_k^2\mathbf{I}\right)\bar\bu_k}
\end{equation}
and
\begin{equation}
\sinr_k^{\textrm{U}}=\frac{q_k\|\bar\bv_k^H\mathbf{H}_k^H\bar{\bu}_k\|^2}{\bar\bv_k^H\left(\sum_{j\neq k}q_j\mathbf{H}_j^H\bar{\bu}_j\bar{\bu}_j^H\mathbf{H}_j+\mathbf{I}\right)\bar\bv_k}.
\end{equation}
Then, fixing the beamformers, the \emph{downlink} power minimization problem can be written as
\begin{equation}\label{eq:downlink}
\begin{split}
&\min_{\{p_k\geq0\}} \sum_{k=1}^K p_k\\
&\textrm{s.t.}~\frac{p_k\|\bar{\bu}_k^H\mathbf{H}_k\bar{\bv}_k\|^2}{\bar{\bu}_k^H\left(\sum_{j\neq k}p_j\mathbf{H}_k\bar{\bv}_j\bar{\bv}_j^H\mathbf{H}_k^H+\sigma_k^2\mathbf{I}\right)\bar{\bu}_k}\geq \gamma_k, \forall k.
\end{split}
\end{equation}
Its dual problem can be obtained by using the Lagrange duality theory with $q_k\ge 0$ corresponding to the Lagrangian multiplier of the $k$th QoS constraint\cite{Codreanu2007}% {\color{blue} [[[add a citation here]]]}:
\begin{equation}\label{eq:down_dual}
\begin{split}
&\max_{\{q_k\geq 0\}} \sum_{k=1}^K \sigma_k^2 q_k\\
&\textrm{s.t.}~\frac{q_k\|\bar{\bv}_k^H\mathbf{H}_k^H\bar{\bu}_k\|^2}{\bar{\bv}_k^H\left(\sum_{j\neq k}q_j\mathbf{H}_j^H\bar{\bu}_j\bar{\bu}_j^H\mathbf{H}_k+\mathbf{I}\right)\bar{\bv}_k}\leq \gamma_k, \forall k.
\end{split}
\end{equation}
It can be shown that problem \eqref{eq:down_dual} is equivalent to the following \emph{uplink} weighted power minimization problem\cite{Song2007}
\begin{equation}\label{eq:uplink}
\begin{split}
&\min_{\{q_k\geq 0\}} \sum_{k=1}^K \sigma_k^2 q_k\\
&\textrm{s.t.}~\frac{q_k\|\bar{\bv}_k^H\mathbf{H}_k^H\bar{\bu}_k\|^2}{\bar{\bv}_k^H\left(\sum_{j\neq k}q_j\mathbf{H}_j^H\bar{\bu}_j\bar{\bu}_j^H\mathbf{H}_k+\mathbf{I}\right)\bar{\bv}_k}\geq \gamma_k, \forall k
\end{split}
\end{equation}
by noting that all the inequality constraints of both problem \eqref{eq:down_dual} and \eqref{eq:uplink} must hold with equality at the optimality and furthermore the corresponding system of linear equations with respect to $\{q_k\}$ has a unique solution\cite[Lemmas 1 \& 2]{Song2007}. To summarize, the classical Lagrange duality theory leads to the well-known uplink-downlink duality theorem:

\begin{theorem} (Uplink-Downlink Duality \cite{Codreanu2007, Song2007}):
\emph{For any given set of normalized beamformers $\{\bar{\bu}_k\}$ and $\{\bar{\bv}_k\}$, a set of given SINR values $\{\gamma_k\}_{k=1}^K$ is achievable in the downlink using the total power consumption $P = \sum_{k=1}^K p_k$ if and only if the same set of SINR values is achievable in the dual uplink channel with the weighted total power of  $P$.}
\end{theorem}

The UDD algorithm is based on the uplink-downlink duality theory. We summarize the UDD algorithm\footnote{As compared to the UDD algorithm in \cite[Algorithm E]{Song2007} which requires updating uplink/downlink power twice at each iteration, the UDD algorithm\cite{Codreanu2007} illustrated in TABLE \ref{tab:pseudo_code_MMSE_SOCP} requires uplink/downlink power update only once at each iteration and thus is more efficient. However, the convergence result to be shown later in Proposition 2 also applies to the UDD algorithm in \cite{Song2007}.} \cite{Codreanu2007} in TABLE \ref{tab:pseudo_code_MMSE_SOCP}. It should be pointed out that, the UDD algorithm requires a feasible initialization \cite{Codreanu2007}. Otherwise, the steps 7 and 11 in the algorithm are not well-defined.\\
It is also worth noting that the work  \cite{Wong2005} (resp. \cite{Codreanu2007}) only shows that the MMSE-SOCP (resp. UDD) algorithm keeps the QCPM objective function nonincreasing as the iteration proceeds. In this paper, we reveal the connection between the MMSE-SOCP and the UDD algorithm. Moreover prove that the two algorithms can monotonically converge to a KKT solution of the QCPM problem.

%\begin{itemize}
%\item[1)]  Fixing the transmit beamformers, update the receive beamformers using the MMSE receiver.
%\item[2)] Update the dual uplink transmit powers by solving the following optimization problem
%\begin{equation}\label{eq:uplink}
%\begin{split}
%&\min_{\{q_k\geq 0\}_k} \sum_{k=1}^K \sigma_k^2 q_k\\
%&\textrm{s.t.}~\frac{q_k\|\bar{\bv}_k^H\mathbf{H}_k^H\bar{\bu}_k\|^2}{\bar{\bv}_k^H
%\left(\sum_{j\neq
%k}q_j\mathbf{H}_j^H\bar{\bu}_j\bar{\bu}_j^H\mathbf{H}_j+\mathbf{I}\right)\bar{\bv}_k}\geq
%\gamma_k, \forall k,
%\end{split}
%\end{equation}
%where $\bar\bu_k \triangleq \frac{\bu_k}{\|\bu_k\|}$ and $\bar\bv_k \triangleq \frac{\bv_k}{\|\bv_k\|}$.
%\item[3)]  Fixing the receive beamformers, update the transmit beamformers using the MMSE receiver in the dual uplink channel.
%\item[4)] Update the transmit powers by solving the following optimization problem
%\begin{equation}\label{eq:downlink}
%\begin{split}
%&\min_{\{p_k\geq0\}_k} \sum_{k=1}^K p_k\\
%&\textrm{s.t.}~\frac{p_k\|\bar{\bu}_k^H\mathbf{H}_k\bar{\bv}_k\|^2}{\bar{\bu}_k^H\left(\sum_{j\neq k}p_j\mathbf{H}_k\bar{\bv}_j\bar{\bv}_j^H\mathbf{H}_k^H+\sigma_k^2\mathbf{I}\right)\bar{\bu}_k}\geq \gamma_k, \forall k
%\end{split}
%\end{equation}
%\end{itemize}

\begin{table}[h]
\centering
\caption{The UDD algorithm}\label{tab:pseudo_code_MMSE_SOCP}
\begin{tabular}{|p{3.3in}|}
\hline
\begin{itemize}
\item \textbf{Input}: $\mathbf{H}_k$, $\sigma_k^2$, $\gamma_k$, $k=1,2,\ldots,K$%\diag[g_{k1}, g_{k2}, \ldots, g_{kL}], \forall k$,  are the frequency--domain channel responses, $\epsilon$ is a predefined small scalar, $T$ is the maximum total number of iterations.
\item \textbf{Output}: the beamformers $\{\bv_k\}$ and $\{\bu_k\}$.
\item [1] \; set $t=0$ and $N\geq1$
\item[2] \; initialize $\bv_k$, $k=1,2,\ldots,K$.
\item [3] \; \textbf{repeat}
\item [4]\; \quad $t \longleftarrow t+1$;
\item [5] \quad  $\hat{\bu}_k\longleftarrow\left(\sum_{j\neq k}\mathbf{H}_k\bv_j\bv_j^H\mathbf{H}_k^H +\sigma_k^2\mathbf{I}\right)^{-1}\mathbf{H}_k\bv_k$, $\forall k$
%\item [6] \quad $\bu_k\longleftarrow\frac{\bar{\bm{u}}_k}{\|\bar{\bm{u}}_k\|}$, $\forall k$
\item [6] \quad $\bar{\bu}_k\longleftarrow\frac{\hat{\bm{u}}_k}{\|\hat{\bm{u}}_k\|}$, $\forall k$
\item [7]  \quad update $q_k$'s by solving \eqref{eq:uplink}~~~//uplink power allocation
\item [8]\quad $\bu_k\longleftarrow\sqrt{q_k}\bar{\bm{u}}_k$, $\forall k$
\item [9] \quad $\hat{\bm{v}}_k\longleftarrow\left(\mathbf{I}+\sum_{j\neq k}\mathbf{H}_j^H\bu_j\bu_j^H\mathbf{H}_j\right)^{-1}\mathbf{H}_k^H\bu_k$, $\forall k$
\item [10] \quad$\bar{\bm{v}}_k\longleftarrow \frac{\hat{\bm{v}}_k}{\|\hat{\bm{v}}_k\|}$, $\forall k$
\item [11] \quad update $p_k$'s by solving \eqref{eq:downlink}~~~//downlink power allocation
\item [12]\quad $\bv_k\longleftarrow\sqrt{p_k}\bar{\bm{v}}_k$, $\forall k$
\item [13]\; \textbf{until}  some convergence criterion is met
%\item [10]\; $\bP_k=\bF\bX_k\bG_k^H, \forall k$
\end{itemize}
\\
\hline
\end{tabular}
\end{table}

\section{The MMSE-DUAL Algorithm}
In this section, we first analyze the KKT conditions of problem (P1). Based on the results of the KKT analysis, we then present a new iterative power minimization algorithm, dubbed MMSE-DUAL, to solve the system of KKT equations. Moreover, the proposed algorithm reveals some connections between  the MMSE-SOCP and the UDD algorithm.

\subsection{KKT Analysis of the QCPM Problem}
First, let us define the Lagrange function associated with problem (P1) as
\begin{equation}
\begin{split}
&\mathcal{L}(\bm{\lambda}, \bv, \bu)\triangleq\sum_{k=1}^K \|\bv_k\|^2\\
&+\sum_{k=1}^K\lambda_k \left(\sum_{j\neq k}|\bu_k^H\mathbf{H}_k\bv_j|^2
+\sigma_k^2\|\bu_k\|^2-\frac{1}{\gamma_k}|\bu_k^H\mathbf{H}_k\bv_k|^2\right)
\end{split}
\end{equation}
where $\bm{\lambda}=[\lambda_1~\lambda_2~\ldots~\lambda_K]^T$ is the Lagrange multiplier vector. For a given optimal primal dual tuple $(\bu,\bv, \bm{\lambda})$, the KKT conditions of problem (P1) are given by

\begin{subequations}\label{eqKKT}
\begin{align}
%&&\nabla_{\bu_k} L
 &\lambda_k\left(\sum_{j\neq k} \mathbf{H}_k\bv_j\bv_j^H\mathbf{H}_k^H +\sigma_k^2\mathbf{I}-\frac{1}{\gamma_k}\mathbf{H}_k\bv_k\bv_k^H\mathbf{H}_k^H\right)\bu_k=0,\;\forall k,\label{neq:KKT1}\\
&\left(\mathbf{I}-\frac{\lambda_k}{\gamma_{k}}\mathbf{H}_k^H\bu_k\bu_k^H\mathbf{H}_k+\sum_{j\neq k}\lambda_j\mathbf{H}_j^H\bu_j\bu_j^H\mathbf{H}_j\right)\bv_k=0,\;\forall k,\label{neq:KKT2}\\
&\lambda_k\left(\sum_{j\neq k}\|\bu_k^H\mathbf{H}_k\bv_j\|^2{+}\sigma_k^2\|\bu_k\|^2{-}\frac{1}{\gamma_k}|\bu_k^H\mathbf{H}_k\bv_k|^2\right){=}0,\;\forall k,\label{neq:KKT3}\\
&\gamma_k\left(\sum_{j\neq k}\|\bu_k^H\mathbf{H}_k\bv_j\|^2+\sigma_k^2\|\bu_k\|^2\right)\leq |\bu_k^H\mathbf{H}_k\bv_k|^2,\;\forall k,\label{neq:KKT4}\\
&\lambda_k\geq 0,\;\forall k,\label{neq:KKT5}
\end{align}
\end{subequations}
where $\eqref{neq:KKT1}$ and  $\eqref{neq:KKT2}$ are the first-order optimality conditions with respect to the receive and transmit beamformers, respectively. The equation $\eqref{neq:KKT3}$ is the complementary condition; and the equations $\eqref{neq:KKT4}$ and $\eqref{neq:KKT5}$ are the primal and dual feasibility conditions. In the sequel, we analyze the above KKT system.
%\end{eqnarray}

%\begin{eqnarray}
%&\nabla_{\bu_k} L = 2\lambda_k\left(\gamma_k\sum_{j\neq k} \mathbf{H}_k\bv_j\bv_j^H\mathbf{H}_k^H +\gamma_k\sigma_k^2\mathbf{I}-\mathbf{H}_k\bv_k\bv_k^H\mathbf{H}_k^H\right)\bu_k=0,\label{neq:KKT1}\\
%&\nabla_{\bv_k} L = 2\left(\mathbf{I}-\lambda_k\mathbf{H}_k^H\bu_k\bu_k^H\mathbf{H}_k+\sum_{j\neq k}\lambda_j\gamma_j\mathbf{H}_j^H\bu_j\bu_j^H\mathbf{H}_j\right)\bv_k=0,\label{neq:KKT2}\\
%&\lambda_k\left(\gamma_k\left(\sum_{j\neq k}\|\bu_k^H\mathbf{H}_k\bv_j\|^2+\sigma_k^2\|\bu_k\|^2\right)-|\bu_k^H\mathbf{H}_k\bv_k|^2\right)=0,\label{neq:KKT3}\\
%&\gamma_k\left(\sum_{j\neq k}\|\bu_k^H\mathbf{H}_k\bv_j\|^2+\sigma_k^2\|\bu_k\|^2\right)\leq |\bu_k^H\mathbf{H}_k\bv_k|^2,\label{neq:KKT4}\\
%&\lambda_k\geq 0, ~k=1,2,\ldots,K. \label{neq:KKT5}
%\end{eqnarray}

\begin{lemma}
\label{lemma:lambdapositive}
\emph{For any primal-dual tuple $(\bu,\bv, \bm{\lambda})$ that satisfies \eqref{eqKKT}, we have $\lambda_k>0, \forall k\in \mathcal{K}$, and
 all the QoS constraints hold with equality.}
\end{lemma}
\begin{proof}
We prove this using contradiction. Assume the contrary that one of the optimal Lagrange multipliers, say $\lambda_k$, equals to zero. Multiplying \eqref{neq:KKT2} by $\bv_k^H$ yields
$$\|\bv_k\|^2+\sum_{j\neq k}\lambda_j\bv_k^H\mathbf{H}_j^H\bu_j\bu_j^H\mathbf{H}_j\bv_k=0.$$
This implies $\bv_k=\bm{0}$, which contradicts the assumption that $\gamma_k>0$.
\end{proof}

Let us define
$\mathbf{C}_k \triangleq \sum_{j\neq k}\mathbf{H}_k\bv_j\bv_j^H\mathbf{H}_k^H +\sigma_k^2\mathbf{I}$
and
$\mathbf{A}_k \triangleq \sum_{j\neq k}\mathbf{H}_k\bv_j\bv_j^H\mathbf{H}_k^H +\sigma_k^2\mathbf{I}-\frac{1}{\gamma_k}\mathbf{H}_k\bv_k\bv_k^H\mathbf{H}_k^H.$
\begin{lemma}
\label{lemma:Akuk}
\emph{
For any primal-dual tuple $(\bu,\bv, \bm{\lambda})$ that satisfies \eqref{eqKKT}, the minimum eigenvalue of $\mathbf{A}_k$ is zero.
Furthermore, the optimal normalized receive beamformers are given by
\begin{align}
\bu_k=\frac{\hat{\bm{u}}_k}{\|\hat{\bm{u}}_k\|}\quad \mbox{with}\quad \hat{\bm{u}}_k=\mathbf{C}_k^{-1}\mathbf{H}_k\bv_k, \;\forall k\in \mathcal{K}.\label{eqMMSEReceiver}
\end{align}
}
\end{lemma}
\begin{proof}
Clearly, \eqref{neq:KKT1} implies that $\mathbf{A}_k$ must have at least one zero eigenvalue. On the other hand, since
\begin{equation}
\label{eq:A_k}
\mathbf{A}_k=\mathbf{C}_k^{\frac{1}{2}}\left(\mathbf{I}-\frac{1}{\gamma_k}\mathbf{C}_k^{-\frac{1}{2}}
\mathbf{H}_k\bv_k\bv_k^H\mathbf{H}_k^H\mathbf{C}_k^{-\frac{1}{2}}\right)\mathbf{C}_k^{\frac{1}{2}},
\end{equation}
$\mathbf{A}_k$ has at most one nonpositive eigenvalue. Hence, the minimum eigenvalue of $\mathbf{A}_k$ is zero. Furthermore, the equation~\eqref{eq:A_k} implies that $\frac{1}{\gamma_k}\bv_k^H\mathbf{H}_k^H\mathbf{C}_k^{-1}\mathbf{H}_k\bv_k=1$ and therefore
\begin{equation}
\label{eq:temp1}
\bH_k \bv_k = \frac{1}{\gamma_k} \bH_k \bv_k \bv_k^H \bH_k^H \mathbf{C}_k^{-1}\bH_k \bv_k.
\end{equation}
Combining \eqref{eq:temp1} and Lemma~\ref{lemma:lambdapositive}, it can be verified that $\bu_k = \mathbf{C}_k^{-1} \bH_k \bv_k$ is the unique solution of \eqref{neq:KKT1} up to scaling.
%otherwise $\mathbf{A}_k$ is of full rank and $\bu_k=\bm{0}$. Since $\lambda_k$ is strictly positive at the optimality, (\ref{neq:KKT1}) is equivalent to the following generalized eigenvalue problem:
%$$\mathbf{C}_k\bu_k=\frac{1}{\gamma_k}\mathbf{H}_k\bv_k\bv_k^H\mathbf{H}_k^H \bu_k.$$
%It is easy to verify that $\bm{f_k}=\mathbf{C}_k^{-1}\mathbf{H}_k\bv_k$ is the unique solution (up to scale) to the above eigenvalue problem.
\end{proof}

%\begin{remark}
% Lemma \ref{lemma:Akuk} verifies that the optimum receiver for the QCPM problem given the transmit beamforming is the MMSE receiver. This result is also shown in \cite{Wong2005}. However, our proof is more rigorous.
% \end{remark}

Defining $\bD_k \triangleq \mathbf{I}+\sum_{j\neq k}\lambda_j\mathbf{H}_j^H\bu_j\bu_j^H\mathbf{H}_j$
and
$\bB_k \triangleq \mathbf{I}+\sum_{j\neq k}\lambda_j\mathbf{H}_j^H\bu_j\bu_j^H\mathbf{H}_j - \frac{\lambda_k}{\gamma_k}\mathbf{H}_k^H\bu_k\bu_k^H\mathbf{H}_k,$
the next lemma follows directly from the KKT equations \eqref{neq:KKT2} and \eqref{neq:KKT3}. The proof of this lemma is similar to the proof of lemma~\ref{lemma:Akuk} and thus omitted from the manuscript.
%
%Similar to the proof of Lemma 3 and notice that the optimal $\bv_k$s must satisfy \eqref{neq:KKT3} with positive $\lambda_k$s, we can prove Lemma 4.

\begin{lemma}
\label{lemma:Bkvk}\it{For any primal-dual tuple $(\bu,\bv, \bm{\lambda})$ that satisfies \eqref{eqKKT}, the minimum eigenvalue of
$\mathbf{B}_k$ is zero. Moreover, the optimal transmit beamformers
are given by
\begin{align}
\bv_k=\mu_k\frac{\hat{\bm{v}}_k}{\|\hat{\bm{v}}_k\|}\quad
\mbox{with}\quad \hat{\bm{v}}_k=\mathbf{D}_k^{-1}\mathbf{H}_k^H\bu_k, \;\forall k \in \mathcal{K}\label{eqTransmitter}
\end{align}
where
the coefficients $\{\mu_k\}$ is chosen such that \eqref{neq:KKT3} is
satisfied.}
\end{lemma}
%
%\begin{proof}
% Note that $\bB_k$ is a positive definite matrix $\bD_k$ minus a rank one matrix $\frac{\lambda_k}{\gamma_k}\mathbf{H}_k^H\bu_k\bu_k^H\mathbf{H}_k$. Hence, it has at most one nonnegative eigenvalue. Moreover, \eqref{neq:KKT2} implies that $\bB_k$ must have at least one zero eigenvalue. Hence, $\bB_k$ has a unique zero eigenvalue and $M-1$ positive eigenvalues, i.e., $\bB_k$ is positive semidefinite.  $\eqref{neq:KKT2}$ is equivalent to the eigenvalue problem
% $$\bD_k\bv_k=\frac{\lambda_k}{\gamma_k}\mathbf{H}_k^H\bu_k\bu_k^H\mathbf{H}_k\bv_k$$
%which has a solution in the form of $\bv_k=\mu_k\frac{\bar{\bm{v}}_k}{\|\bar{\bm{v}}_k\|}$ with $\bar{\bm{v}}_k=\bD_k^{-1}\mathbf{H}_k^H\bu_k$. Due to $\eqref{neq:KKT3}$, $\mu_k$s must satisfy Eq. \eqref{neq:KKT3}.
%\end{proof}
\begin{corollary} \emph{For any primal-dual tuple $(\bu,\bv, \bm{\lambda})$ that satisfies \eqref{eqKKT}, the Lagrange
multipliers $\{{\lambda}_k\}$  satisfies the system of equations
\begin{equation}\label{eq:standard_fun_lambda}
\lambda_k=\frac{\gamma_k}{1+\gamma_k}\frac{1}{\bu_k^H\bH_k\bUpsilon\bH_k^H\bu_k},\;\forall k\in\mathcal{K}.
\end{equation}
where $\bUpsilon\triangleq \left(\mathbf{I}+\sum_{j=1}^K\lambda_j\mathbf{H}_j^H\bu_j\bu_j^H\mathbf{H}_j\right)^{-1}$. Moreover, the unique solution of \eqref{eq:standard_fun_lambda} can be obtained by a fixed point iteration.}
\end{corollary}
\begin{proof}
From Lemma~\ref{lemma:Bkvk}, we have $\bB_k=\bD_k-\frac{\lambda_k}{\gamma_k}\mathbf{H}_k^H\bu_k\bu_k^H\mathbf{H}_k\succeq 0.$
It follows from Schur complement that
\begin{equation}\label{eq:schurc}
\left[\begin{array}{cc}
 \bD_k &\sqrt{\frac{\lambda_k}{\gamma_k}}\mathbf{H}_k^H\bu_k\\
\sqrt{\frac{\lambda_k}{\gamma_k}}\bu_k^H\mathbf{H}_k & 1
\end{array}\right]\succeq 0.
\end{equation}
Since $\bD_k$ is positive definite, using Schur complement again implies
$$\frac{\lambda_k}{\gamma_k}\bu_k^H\mathbf{H}_k\bD_k^{-1}\mathbf{H}_k^H\bu_k \leq 1.$$
For any primal-dual tuple $(\bu,\bv, \bm{\lambda})$ that satisfies the KKT condition \eqref{eqKKT}, the above inequality must hold with equality;
otherwise $\bB_k\succ 0$ by the Schur complement, which contradicts
Lemma~\ref{lemma:Bkvk}. Hence, it holds that
\begin{equation}\label{eq:lambda_eq1}
\lambda_k=\frac{\gamma_k}{\bu_k^H\mathbf{H}_k\left(\bUpsilon^{-1}-\lambda_k\mathbf{H}_k^H\bu_k\bu_k^H\mathbf{H}_k\right)^{-1}\mathbf{H}_k^H\bu_k}, \;\forall k\in\mathcal{K}.
\end{equation}
On the other hand,
\begin{equation}
\begin{split}
&\bu_k^H\mathbf{H}_k\left(\bUpsilon^{-1}-\lambda_k\mathbf{H}_k^H\bu_k\bu_k^H\mathbf{H}_k\right)^{-1}\mathbf{H}_k^H\bu_k\\
=&\bu_k^H\mathbf{H}_k\left(\bI-\lambda_k\bUpsilon\mathbf{H}_k^H\bu_k\bu_k^H\mathbf{H}_k\right)^{-1}\bUpsilon\mathbf{H}_k^H\bu_k\\
=&\bu_k^H\mathbf{H}_k\bUpsilon\mathbf{H}_k^H\bu_k\left(1-\lambda_k\bu_k^H\mathbf{H}_k\bUpsilon\mathbf{H}_k^H\bu_k\right)^{-1},
\end{split}
\end{equation}
where the second equality is due to the identity $(\bI+\bA\bB)^{-1}\bA=\bA(\bI+\bB\bA)^{-1}$\cite[Sec. 3.2.4]{Mtx_book}. Hence \eqref{eq:lambda_eq1} can be rewritten as
\begin{equation}
\lambda_k=\gamma_k\left(\frac{1}{\bu_k^H\bH_k\bUpsilon\bH_k^H\bu_k}-\lambda_k\right),\;\forall k\in\mathcal{K},
\end{equation}
which implies \eqref{eq:standard_fun_lambda}.
Moreover, the right hand side of \eqref{eq:standard_fun_lambda} is a standard
function\footnote{A vector function $\bm{f}(\bm{\lambda})$ is a standard function if it satisfies 1) $\bm{f}(\bm{\lambda})>0$; 2) $\bm{f}(\bm{\lambda})\geq \bm{f}(\bm{\lambda}')$ for $\bm{\lambda}\geq \bm{\lambda}'$; 3) $\alpha\bm{f}(\bm{\lambda})\geq \bm{f}(\alpha\bm{\lambda})$ for $\alpha>1$. If $\bm{f}(\bm{\lambda})$ is a standard function, the system of equations $\bm{\lambda}=\bm{f}(\bm{\lambda})$ has a unique solution which can be obtained by a fixed point algorithm. See more details in \cite{Yates1995}.} of $\{\lambda_k\}$. Hence, the solution of \eqref{eq:standard_fun_lambda} is
unique and can be obtained by a fixed point algorithm \cite{Yates1995}.
\end{proof}

%Following Corollary 1, we have the following lemma which will be used in convergence analysis.
%\begin{lemma}
%The optimal solution to problem \eqref{eq:P1-socp} is unique up to phase rotation.
%\end{lemma}
%\begin{proof}
%Note that \eqref{neq:KKT2}-\eqref{neq:KKT5} are just the KKT system for problem \eqref{eq:P1-socp}.
%According to Corollary 1, the set of Lagrange multipliers $\{\lambda_k\}$ is unique. Hence, $\{\bB_k\}$ is unique for fixed $\bu_k$'s. Moreover, since each $\bB_k$ is positive semidefinite with unique zero eigenvalue, we conclude that $\{\bv_k\}$ satisfying \eqref{neq:KKT2} is unique up to phase rotation.
%\end{proof}
\subsection{The Proposed MMSE-DUAL Algorithm And Its Relations With The Existing Algorithms}
Based on the results of the KKT analysis, we here present a new iterative power minimization algorithm, dubbed MMSE-DUAL algorithm, to
solve the system of KKT equations (\ref{neq:KKT1}-\ref{neq:KKT5}). Our proposed MMSE-DUAL algorithm first updates the receive beamformers using the MMSE receiver \eqref{eqMMSEReceiver}, followed by the update of the transmit beamformers using equations  \eqref{eqTransmitter} and \eqref{eq:standard_fun_lambda}. The algorithm is outlined in TABLE~\ref{tbl:pseudo_code_MIMO}. In this table, $N$ denotes the total number of fixed point iterations for calculating the optimal Lagrange multipliers.

Before stating the properties of the proposed algorithm, let us first see how the MMSE-DUAL algorithm plays a key role in establishing the connection between the MMSE-SOCP algorithm and the UDD algorithm. First notice that in the MMSE-DUAL algorithm, the procedure of updating the transmit beamformers (i.e., Steps 7-13 in TABLE II) is equivalent to solving
\begin{equation*}
\label{eq:P2}
\textrm{(P2)}\begin{split}
&\min_{\{\bv_k\}} \;\;\sum_{k=1}^K \|\bv_k\|^2\\
&\textrm{s.t.}~~\frac{|\bu_k^H\mathbf{H}_k\bv_k|^2}{\sum_{j\neq k}|\bu_k^H\mathbf{H}_k\bv_j|^2{+}\sigma_k^2\|\bu_k\|^2}{\geq} \gamma_k, \;\forall k \in \mathcal{K}.
\end{split}
\end{equation*}
This follows from the fact that the KKT condition of (P2) are identical to (\ref{neq:KKT2}-\ref{neq:KKT5}) and strong duality holds for problem (P2)\cite{Ottersten2001}. Since the updates of the receive beamformer in MMSE-DUAL is the same as the receiver update in the MMSE-SOCP algorithm, the MMSE-DUAL algorithm is in essence the MMSE-SOCP algorithm. The only difference of the MMSE-DUAL algorithm with the MMSE--SOCP \cite{Wong2005} is that, instead of updating the transmit beamformers by directly solving the SOCP (P2), we use semi-closed form computation  \eqref{eqTransmitter} and \eqref{eq:standard_fun_lambda}. It is not hard to see that\footnote{In this paper, we consider practical cases for complexity comparison, i.e., when the number of transmit antennas is greater than the number of antennas at each receiver.} the complexity of each iteration of the fixed point algorithm is dominated by the computation of $\bUpsilon^{-1}$, which is $\mathcal{O}(KM^2+M^3)$. Hence, the MMSE-DUAL algorithm has lower complexity than the MMSE-SOCP algorithm.

Next we explore the relation between the MMSE-DUAL and the UDD algorithm. Comparing the algorithms in TABLE I and II, it is not difficult to see that, the dual variables $\{\lambda_k\}$ and the auxiliary
variables $\{\mu_k\}$ in the MMSE-DUAL algorithm respectively play the role of the dual uplink
transmit power and the downlink transmit power in the UDD
algorithm. This implies that, although the UDD algorithm is
developed from a different point of view, it works in a similar way
as the MMSE-DUAL algorithm towards solving the KKT system \eqref{eqKKT}. Furthermore, the constraints of problem \eqref{eq:downlink} must be satisfied with equality at the optimality. Hence it is readily seen that the MMSE-DUAL algorithm updates the auxiliary variables
$\{\mu_k\}$ in the same way as the downlink transmission power
$\{p_k\}$ in the UDD algorithm; see Step 11 in TABLE I and Step 13 in TABLE II. The only difference between the two
algorithms lies in the update of the dual variables $\{\lambda_k\}$
(or equivalently the dual uplink transmit power $q_k$ in UDD).
In the MMSE-DUAL algorithm, the Lagrange multipliers $\{\lambda_k\}$
are updated by the fixed point algorithm independent of the current
transmit beamformers, while the update of uplink transmit power
in the UDD algorithm depends on the current transmission
beamformers; see Steps 7-10 in TABLE II and Step 7 in TABLE I. Note that, if problem \eqref{eq:uplink} is feasible, the update of the uplink transmit power
in the UDD algorithm is equivalent to solving a linear system of $\{q_k\}$ which has complexity of $\mathcal{O}(K^2M+K^3)$. Hence, when $M\approx K$, the UDD algorithm and the MMSE-DUAL algorithm have comparable complexity. But in the general single stream case where $M>K$, the UDD algorithm has lower complexity than the MMSE-DUAL algorithm.

\begin{table}[h]
\centering
\caption{The MMSE-DUAL algorithm}\label{tbl:pseudo_code_MIMO}
\begin{tabular}{|p{3.6in}|}
\hline
\begin{itemize}
\item \textbf{Input}: $\mathbf{H}_k$, $\sigma_k^2$, $\gamma_k$, $k=1,2,\ldots,K$%\diag[g_{k1}, g_{k2}, \ldots, g_{kL}], \forall k$,  are the frequency--domain channel responses, $\epsilon$ is a predefined small scalar, $T$ is the maximum total number of iterations.
\item \textbf{Output}: the beamformers $\{\bv_k\}$ and $\{\bu_k\}$.
\item [1] \; set $t=0$ and $N\geq1$
\item[2] \; initialize $\bv_k$ and $\lambda_k$, $k=1,2,\ldots,K$.
\item [3] \; \textbf{repeat}
\item [4]\; \quad $t \longleftarrow t+1$;
\item [5] \quad  $\hat{\bm{u}}_k\longleftarrow\left(\sum_{j\neq k}\mathbf{H}_k\bv_j\bv_j^H\mathbf{H}_k^H +\sigma_k^2\mathbf{I}\right)^{-1}\mathbf{H}_k\bv_k$, $\forall k$
\item [6] \quad $\bu_k\longleftarrow\frac{\hat{\bm{u}}_k}{\|\hat{\bm{u}}_k\|}$, $\forall k$
\item [7] \quad \textbf{for} $n=1$ to $N$\quad//fixed point algorithm
\item [8] \quad\quad\;$\bUpsilon\longleftarrow \left(\mathbf{I}+\sum_{j=1}^K\lambda_j\mathbf{H}_j^H\bu_j\bu_j^H\mathbf{H}_j\right)^{-1}$
\item [9]  \quad\quad\;$\lambda_k\longleftarrow\frac{\gamma_k}{1+\gamma_k}\frac{1}{\bu_k^H\bH_k\bUpsilon\bH_k^H\bu_k}$, $\forall k$
    \item[10] \quad\textbf{end for}
\item [11] \quad $\hat{\bm{v}}_k\longleftarrow\left(\mathbf{I}+\sum_{j\neq k}\lambda_j\mathbf{H}_j^H\bu_j\bu_j^H\mathbf{H}_j\right)^{-1}\mathbf{H}_k^H\bu_k$, $\forall k$
\item [12] \quad$\bar{\bm{v}}_k\longleftarrow \frac{\hat{\bm{v}}_k}{\|\hat{\bm{v}}_k\|}$, $\forall k$
\item [13] \quad solve for the linear system of $\{\mu_k\}$:
\item []\quad\quad $\frac{1}{\gamma_k}\mu_k|\bu_k^H\mathbf{H}_k\bar{\bm{v}}_k|^2-\sum_{j\neq k}\mu_j\|\bu_k^H\mathbf{H}_k\bar{\bm{v}}_j\|^2=\sigma_k^2\|\bu_k\|^2$, $\forall k$
\item [14]\quad $\bv_k\longleftarrow\sqrt{\mu_k}\bar{\bm{v}}_k$, $\forall k$
\item [15]\; \textbf{until} some convergence criterion is met
%\item [10]\; $\bP_k=\bF\bX_k\bG_k^H, \forall k$
\end{itemize}
\\
\hline
\end{tabular}
\end{table}

In a nutshell, similar to the MMSE-DUAL algorithm, the MMSE-SOCP algorithm and the UDD algorithm both work towards solving the KKT system \eqref{eqKKT}. In the next section, we rigorously show that the three algorithms indeed monotonically converge to a KKT solution.
%All the above comparison leads to a simple fact: if we can show the convergence of the MMSE-DUAL algorithm, then we can immediately obtain the convergence of the UDD and the MMSE-SOCP algorithm as well. Unfortunately, the convergence of the MMSE-DUAL algorithm is not easy to obtain, at least it does not follow directly from the classical theory of alternating optimization methods \cite{Bertsekas_book}. This is due to the coupling of $\bv$ and $\bu$ in the constraints of the original QCPM problem (P1). In the next subsection, we rigorously show that the MMSE-DUAL algorithm indeed converges to a KKT solution, by exploring certain special structure of problem (P1).

\section{Convergence Results of the MMSE-DUAL/MMSE-SOCP/UDD Algorithms}\label{subConvergence}
In this section, we establish the convergence of the three algorithms. Compared to the existing convergence results of the MMSE-SOCP algorithm in \cite{Wong2005} and the UDD algorithm in \cite{Codreanu2007},
%which only shows the monotonicity of the objective function (i.e., the power consumption non-increases as the iteration proceeds),
we show below a stronger convergence result that every limit point of the three algorithms is a KKT solution of \eqref{eqKKT}. Before stating the convergence results, we first present three lemmas which will be used later in the convergence proof.

%The following lemma illustrates a property of problem \eqref{eq:P1-socp}.
\begin{lemma}
\label{lem:phase-rotation}
For any feasible $\{\bu_k\}$, the optimal solution to problem \eqref{eq:P1-socp} is unique up to phase rotation.
\end{lemma}
 \begin{proof}
 Let $\{\bv_k^*\}$ and $\{\bv_k^{**}\}$ be any two optimal solutions to problem \eqref{eq:P1-socp}. Moreover, we define $\theta_k^*\triangleq\angle\left(\bu_k^H\bH_k\bv_k^*\right)$ and $\theta_k^{**}\triangleq\angle\left(\bu_k^H\bH_k\bv_k^{**}\right)$, $\forall k\in \mathcal{K}$. In the following, we prove $\bv_k^*=\bv_k^{**}e^{j(\theta_k^*-\theta_k^{**})}$, $\forall k\in \mathcal{K}$.

It is noted that any feasible solution to the following SOCP
  \begin{equation}
\label{eq:P1-socp-star}
\begin{split}
\min_{\bv} &\sum_{k=1}^K \|\bv_k\|^2\\
\textrm{s.t.}\;\; &e^{-j\theta_k^*}\frac{1}{\sqrt{\gamma_k}}\bu_k^H\mathbf{H}_k\bv_k{\geq} \sqrt{\sum_{j\neq k}|\bu_k^H\mathbf{H}_k\bv_j|^2{+}\sigma_k^2\|\bu_k\|^2}, \;\forall k\in \mathcal{K},\\
&\textrm{Im}(e^{-j\theta_k^*}\bu_k^H\mathbf{H}_k\bv_k)=0, \;\forall k\in \mathcal{K}
\end{split}
\end{equation}
is feasible to problem \eqref{eq:P1-socp}. It follows that the optimal value of problem \eqref{eq:P1-socp} is not greater than that of problem \eqref{eq:P1-socp-star}. On the other hand, it is easily seen that the optimal solution $\{\bv_k^*\}$ to problem \eqref{eq:P1-socp} is a feasible solution to problem \eqref{eq:P1-socp-star}. Hence, $\{\bv_k^*\}$ is an optimal solution to problem \eqref{eq:P1-socp-star}. Similarly, we infer that $\{\bv_k^{**}\}$ is an optimal solution to the following SOCP
  \begin{equation}
\label{eq:P1-socp-starstar}
\begin{split}
\min_{\bv} &\sum_{k=1}^K \|\bv_k\|^2\\
\textrm{s.t.}\;\; &e^{-j\theta_k^{**}}\frac{1}{\sqrt{\gamma_k}}\bu_k^H\mathbf{H}_k\bv_k{\geq} \sqrt{\sum_{j\neq k}|\bu_k^H\mathbf{H}_k\bv_j|^2{+}\sigma_k^2\|\bu_k\|^2}, \;\forall k\in \mathcal{K},\\
&\textrm{Im}(e^{-j\theta_k^{**}}\bu_k^H\mathbf{H}_k\bv_k)=0, \;\forall k\in \mathcal{K}.
\end{split}
\end{equation}
By comparing the above two SOCP formulations and noting $\sum_{k=1}^K \Vert\bv_k^*\Vert^2=\sum_{k=1}^K \Vert\bv_k^{**}\Vert^2$, we conclude that $\{\bv_k^{**}e^{j(\theta_k^*-\theta_k^{**})}\}$ is an optimal solution to problem \eqref{eq:P1-socp-star}. Since the SOCP problem \eqref{eq:P1-socp-star} has a strictly convex objective function, it has a unique solution. Thus we have $\bv_k^*=\bv_k^{**}e^{j(\theta_k^*-\theta_k^{**})}$, $\forall k\in \mathcal{K}$, implying that the optimal solution to problem \eqref{eq:P1-socp} is unique up to phase rotation. This completes the proof.
 \end{proof}

\begin{lemma}\label{lem:lemma4-prop1}
If $\{\bu_k\}$ and $\{\bv_k\}$ satisfy the SINR constraints of problem (P1) with equality and for each $k$ $\bu_k$ is an MMSE-receiver, i.e., $\bu_k=\alpha_k\bC_k^{-1}\bH_k\bv_k$ for all $k$ with $\alpha_k$ being an arbitrary nonzero constant, then we have
\begin{equation}\label{eq:secIV_lemma_eq}
\left(\bC_k-\frac{1}{\gamma_k}\bH_k\bv_k\bv_k^H\bH_k^H\right)\bu_k=0,\;\forall k.
\end{equation}
%where $\bC_k$'s are defined as in Section III.
\end{lemma}
\begin{proof}
Since $\{\bu_k\}$ and $\{\bv_k\}$ satisfy the constraints of problem (P1) with equality, we have
$$\sum_{j\neq k}\|\bu_k^H\mathbf{H}_k\bv_j\|^2{+}\sigma_k^2\|\bu_k\|^2-\frac{1}{\gamma_k}|\bu_k^H \mathbf{H}_k\bv_k|^2{=}0,\;\;\forall k$$
which can be rearranged as
\begin{equation}
\label{eq:bk2}
\bu_k^H\left(\mathbf{C}_k-\frac{1}{\gamma_k}\mathbf{H}_k\bv_k\bv_k^H\mathbf{H}_k^H\right)\bu_k=0, \;\forall k.
\end{equation}
By substituting the MMSE receiver $\bu_k=\alpha_k\bC_k^{-1}\bH_k\bv_k$ into \eqref{eq:bk2} and noting that $\alpha_k\neq 0$, $\forall k$, we get
\begin{equation}\nonumber
\bv_k^H\mathbf{H}_k^H\mathbf{C}_k^{-1}\mathbf{H}_k\bv_k\left(1-\frac{1}{\gamma_k}\bv_k^H\mathbf{H}_k^H\mathbf{C}_k^{-1}\mathbf{H}_k\bv_k\right)=0, \;\forall k.
\end{equation}
Since we have $\bv_k^H\mathbf{H}_k^H\mathbf{C}_k^{-1}\mathbf{H}_k\bv_k\neq 0$ (due to $\mathbf{H}_k\bv_k\neq 0$), it follows that
\begin{equation}
\label{eq:proptemp7}
\frac{1}{\gamma_k}\bv_k^H\mathbf{H}_k^H\mathbf{C}_k^{-1} \mathbf{H}_k\bv_k=1, \;\forall k.
\end{equation}
This implies that for each $k$ the matrix $\mathbf{C}_k-\frac{1}{\gamma_k}\mathbf{H}_k\bv_k\bv_k^H \mathbf{H}_k^H$ is positive semidefinite. Hence, from \eqref{eq:bk2}, we obtain
\begin{equation}
\label{eq:fk1}
\left(\mathbf{C}_k-\frac{1}{\gamma_k}\mathbf{H}_k\bv_k\bv_k^H \mathbf{H}_k^H\right)^{\frac{1}{2}}\bu_k=\bm{0}, \;\forall k,
\end{equation}
which completes the proof.
\end{proof}

Due to the uplink--downlink duality, we also have the following lemma. The proof of the lemma is similar to that of Lemma \ref{lem:lemma4-prop1} and thus omitted for brevity.
\begin{lemma}\label{lem:lemma5-prop2}
If $\{\bu_k\}$ and $\{\bv_k\}$ satisfy the SINR constraints of the virtual uplink weighted power minimization problem\footnote{Note that the SINR constraints in \eqref{eq:uplink} hold true for $\bv_k$'s.} \eqref{eq:uplink} with equality and moreover for each $k$, $\bv_k$ is a (virtual uplink) MMSE-receiver, i.e., $\bv_k=\beta_k\bD_k^{-1}\bH_k^H\bu_k$ with $\beta_k$ being an arbitrary nonzero constant, then we have
\begin{equation}\label{eq:secIV_lemma_eq}
\left(\bD_k-\frac{1}{\gamma_k}\bH_k^H\bu_k\bu_k^H\bH_k\right)\bv_k=0,\;\forall k,
\end{equation}
where $\bD_k$'s are defined as in Section III.
\end{lemma}

\begin{remark}
\label{rem:remak-secIV}
Lemmas \ref{lem:lemma4-prop1} and \ref{lem:lemma5-prop2} indicate that, for the downlink/virtual uplink power minimization problem, if a set of transmit and receive beamformers satisfies the SINR constraints with equality and moreover complies with the MMSE receiver structure, the first--order optimality condition with respect to the downlink/virtual uplink receive beamformers follows. In other words, if a set of transmit and receive beamformers fulfills both the assumptions of Lemma \ref{lem:lemma4-prop1} and  Lemma \ref{lem:lemma5-prop2}, then it is a KKT solution to the KKT system \eqref{eqKKT}.
\end{remark}
\subsection{Convergence of the MMSE-SOCP/MMSE-DUAL Algorithms}
%It can be observed that, given a set of transmit and receive beamformers, if the transmit beamformers is the optimal solution to problem (P1) with fixed receive beamformers, then this set of transmit and receive beamformers also satisfies the first--order optimality condition of problem (P1) with respect to the transmit beamformers. Based on this observation and the remark on Lemma \ref{lem:lemma4-prop1} in Remark \ref{rem:remak-secIV}, we prove the convergence result of the MMSE-SOCP/MMSE-DUAL algorithm by showing that 1) every limit point of the MMSE-SOCP/MMSE-DUAL algorithm, denoted by $\{\bu_k^*, \bv_k^*\}$, satisfies the SINR constraints of problem (P1) with equality and moreover fulfills the MMSE-receiver structure, and that 2) $\{\bv_k^*\}$ is the optimal solution to problem (P1) with there $\bu_k$'s fixed as $\bu_k^*$'s.
In this subsection, we study the convergence behavior of the MMSE-SOCP (or equivalently MMSE-DUAL) algorithm.

 \begin{prop}\label{prop1-secIV}
 \it{Let $\{(\bv^r,\bu^r,\boldmath{\lambda}^r)\}_{r=1}^{\infty}$ denote a sequence generated by the MMSE-DUAL algorithm (or equivalently the MMSE-SOCP algorithm). Suppose $(\bv^0,\bu^0,\boldmath{\lambda}^0)$ is feasible for problem (P1), then every limit point of $\{(\bv^r,\bu^r,\boldmath{\lambda}^r)\}_{r=1}^{\infty}$ is a KKT point of (P1).}
 \end{prop}
\begin{proof}
Here we prove the result for the MMSE-SOCP algorithm. The proof for the MMSE-DUAL algorithm follows immediately due to its equivalence to MMSE-SOCP algorithm. %Let $\bu^r \triangleq \{\bu_k^r| k \in \mathcal{K}\}$ be the set of receive beamformers obtained at iteration $r$. Similarly, define $\bv^r \triangleq \{\bv_k^r| k \in \mathcal{K}\}$ to be the set of transmit beamformers at iteration $r$ in the MMSE-SOCP algorithm.
The iterations of the MMSE-SOCP algorithm are illustrated as $\bv^{r-1}\rightarrow\bu^{r}\rightarrow\bv^{r}$, where the two arrows correspond to the two update rules shown in Sec. II.B in order. Since the objective function is coercive, $\{\bv_k^r\}$ is bounded and consequently $\{\bu_k^r\}$ is bounded as well. Hence, the sequence $\{(\bu^{r},\bv^{r})\}$ has at least one limit point. Consider a subsequence $\{(\bu^{r_j},\bv^{r_j})\}_{j=1}^{\infty}$ converging to the limit point $\{\bu^*,\bv^*\}$. Moreover, by further restricting to another subsequence, we can assume that $\bv^{r_j-1}$ converges to a limit point $\bv^{**}$.

In the sequel, we first prove that $\bv^* = \bv^{**}$ (up to a phase rotation). Clearly,
\begin{equation}\label{eq:proptemp1}
\sum_{k=1}^K \|\bv_k^*\|^2 = \sum_{k=1}^K \|\bv_k^{**}\|^2
\end{equation}
since the objective value is decreasing and it is bounded from below. Now consider a fixed transmit beamformer $\bv$ so that
\begin{equation}
\label{eq:proptemp2}
{\rm SINR}_k (\bv,\bu_k^*) > \gamma_k, \forall k \in \mathcal{K},
\end{equation}
where ${\rm SINR}_k(\bv,\bu_k^*)$ is defined in \eqref{eq:SINRk}. Due to the continuity of the SINR function, there exists an index $i$ so that for all $j>i$, $\sinr_k(\bv,\bu_k^{r_j}) \geq \gamma_k, \forall k \in \mathcal{K}$. Since at each iteration of the algorithm, the transmit beamformers are updated  after the receive beamformers, for all $j>i$ we have $
\sum_{k=1}^K \|\bv_k^{r_j}\|^2 \leq \sum_{k=1}^K \|\bv_k\|^2.
$
Letting $j \rightarrow \infty$ implies
\begin{equation}
\label{eq:proptemp3}
\sum_{k=1}^K \|\bv_k^*\|^2 \leq \sum_{k=1}^K \|\bv_k\|^2.
\end{equation}
Furthermore, according to the update rule, we have $\sinr_k (\bv^{r_j},\bu_k^{r_j}) \geq \gamma_k, \forall k$ and thus
\begin{equation}
\label{eq:proptemp4}
\sinr_k(\bv^*,\bu_k^*) \geq \gamma_k,\;\forall k\in \mathcal{K}.
\end{equation}
Combining \eqref{eq:proptemp4} with the fact that  \eqref{eq:proptemp3} holds for any $\bv$ satisfying \eqref{eq:proptemp2}, we obtain\footnote{Note that, for any $\bv$ such that ${\rm SINR}_k (\bv,\bu_k^*) \geq \gamma_k, \forall k \in \mathcal{K}$, we can scale up $\bv$ with any constant $s>1$ so that the scaled $\bv$ (i.e., $s\bv$) satisfies \eqref{eq:proptemp2}. Furthermore, analogous to \eqref{eq:proptemp3}, we have $\sum_{k=1}^K \|\bv_k^*\|^2 \leq \sum_{k=1}^K s^2 \|\bv_k\|^2$, $\forall s>1$, implying $\sum_{k=1}^K \|\bv_k^*\|^2 \leq \sum_{k=1}^K \|\bv_k\|^2$. Therefore, combining \eqref{eq:proptemp4} with the fact that  \eqref{eq:proptemp3} holds for any $\bv$ satisfying \eqref{eq:proptemp2} implies \eqref{eq:proptemp5}. Similar arguments are also used in the proof of Proposition 2.}
\begin{equation}
\label{eq:proptemp5}
\begin{split}
\bv^* \in {\rm arg}\min_{\bv} \; &\sum_{k=1}^K \|\bv_k\|^2 \\
\textrm{s.t.}\;\;\; &\sinr_k(\bv,\bu_k^*) \geq \gamma_k, \forall k.
\end{split}
\end{equation}
On the other hand, since the update of receive beaformers using MMSE receivers keeps the SINR feasibility, we have $\sinr_k(\bv^{r_j-1},\bu_k^{r_j}) \geq \gamma_k$.  Letting $j\rightarrow \infty$, we obtain \begin{equation}
\label{eq:proptemp6}
\sinr_k(\bv^{**},\bu_k^*) \geq \gamma_k, \;\forall k.
\end{equation}
Combining \eqref{eq:proptemp5}, \eqref{eq:proptemp6} and \eqref{eq:proptemp1}, we infer that $\bv^{**}$ is also an optimal solution to problem \eqref{eq:proptemp5}. Hence, according to Lemma \ref{lem:phase-rotation}, we have $\bv^* = \bv^{**}$ up to an appropriate phase rotation.
%\begin{equation}
%\nonumber
%\begin{split}
%\bv^{**} \in{\rm arg}\min_{\bv} \; &\sum_{k=1}^K \|\bv_k\|^2\\
%\textrm{s.t.}\;\;\; &\sinr_k(\bv,\bu_k^*) \geq \gamma_k, \forall k,
%\end{split}
%\end{equation}
%which together with \eqref{eq:proptemp5} imply that $\bv^* = \bv^{**}$ after doing the appropriate phase %rotation according to Lemma \ref{lem:phase-rotation}.

Next, with $\bv^*=\bv^{**}$, we prove that the limit point $(\bu^*,\bv^*)$ is a KKT point of (P1). Based on the receive beamformer update rule in the algorithm,  $\bu_k^{r_j} = (\mathbf{C}_k^{r_j-1})^{-1} \bH_k \bv_k^{r_j-1}$, where $\mathbf{C}_k^{r_j - 1} := \sum_{\ell\neq k} \bH_k \bv_\ell^{r_j-1}(\bv_\ell^{r_j-1})^H \bH_k^H + \sigma_k^2\mathbf{I}$. Letting $j \rightarrow \infty$ implies
\begin{equation}
\label{eq:propOptu}
\bu_k^{*} = (\mathbf{C}_k^{*})^{-1} \bH_k \bv_k^{*}\;\forall k
\end{equation}
with $\mathbf{C}_k^*=\sum_{j\neq
k}\mathbf{H}_k\bv_j^*(\bv_j^*)^H\mathbf{H}_k^H
+\sigma_k^2\mathbf{I}$. On the other hand, \eqref{eq:proptemp5} implies that there exists a set of multipliers $\lambda_k^* \geq 0$ so that
\begin{align}
%&\nabla_{\bu_k} L = 2\lambda_k\left(\gamma_k\sum_{j\neq k} \mathbf{H}_k\bv_j\bv_j^H\mathbf{H}_k^H +\gamma_k\sigma_k^2\mathbf{I}-\mathbf{H}_k\bv_k\bv_k^H\mathbf{H}_k^H\right)\bu_k=0,\label{neq:KKT1}\\
&\left(\mathbf{I}{-}\frac{\lambda_k^*}{\gamma_k}\mathbf{H}_k^H\bu_k^*(\bu_k^*)^H\mathbf{H}_k{+}\sum_{j\neq k}\lambda_j^*\mathbf{H}_j^H\bu_j^*(\bu_j^*)^H\mathbf{H}_j\right)\bv_k^*{=}0,\;\forall k,\label{eq:bk0}\\
&\gamma_k\left(\sum_{j\neq k}\|(\bu_k^*)^H\mathbf{H}_k\bv_j^*\|^2{+}\sigma_k^2\|\bu^*_k\|^2\right){-}|(\bu_k^*)^H \mathbf{H}_k\bv_k^*|^2{=}0,\;\forall k.\label{eq:bk1}
\end{align}
By Lemma \ref{lem:lemma4-prop1}, we infer from  \eqref{eq:propOptu} and \eqref{eq:bk1} that
\begin{equation}
\label{eq:fk1}
\left(\mathbf{C}_k^*-\frac{1}{\gamma_k}\mathbf{H}_k\bv_k^*(\bv_k^*)^H \mathbf{H}_k^H\right)\bu_k^*=\bm{0},\;\forall k.
\end{equation}
Clearly, the equations \eqref{eq:bk0}, \eqref{eq:bk1} and \eqref{eq:fk1} imply
that  the limit point $(\bu^*,\bv^*)$ is a KKT point of (P1).
\end{proof}

\begin{figure}[t]
\centering
\includegraphics[width=4.5in]{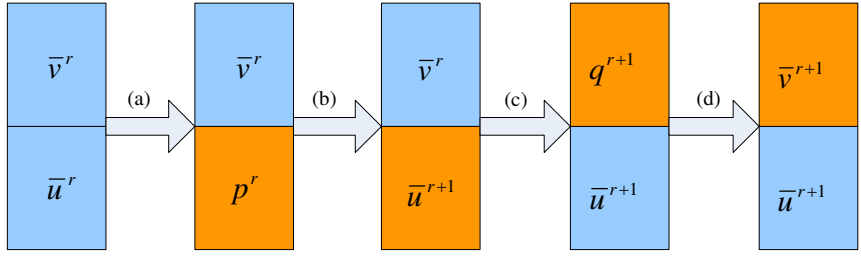}%{flow_chart_UDD.eps}
\caption{The flow chart of the UDD algorithm.}
\label{fig:udd-flow-chart}
\end{figure}
\subsection{Convergence of the UDD Algorithm}
%Based on Remark \ref{rem:remak-secIV}, we prove the convergence result of the UDD algorithm by showing that every limit point of the UDD algorithm satisfy with equality the SINR constraints of both the downlink and virtual uplink power minimization problem, and moreover fulfills both the downlink and uplink MMSE receiver structure.
The following theorem states our convergence result of the UDD algorithm:

\begin{prop} \emph{Let $\{(\bar{\bv}^r, {p}^r,\bar{\bu}^r,{q}^r)\}_{r=1}^{\infty}$ denote a sequence generated by the UDD algorithm. Suppose $(\bar{\bv}^0,{p}^0,\bar{\bu}^0,{q}^0)$ is feasible for problem (P1), then every limit point of $\{(\bar{\bv}^r,{p}^r,\bar{\bu}^r,{q}^r)\}_{r=1}^{\infty}$ is a KKT point of (P1).}
\end{prop}
\begin{proof}
For the ease of understanding, the iterations of the UDD algorithm are visually presented in Fig. \ref{fig:udd-flow-chart} where each arrow indicates an update rule as labeled. Let $$\{\bar{\bv}^r, p^r, \bar{\bu}^r, q^r\}\triangleq\left\{\{\bar{\bv}_k^r\}_{k=1}^K, \{p^r_k\}_{k=1}^K, \{\bar{\bu}_k^r\}_{k=1}^K, \{q^r_k\}_{k=1}^K\right\}$$ be the sequence generated by the UDD algorithm. Clearly,  the sequences $\{\bar{\bv}^{r}\}$ and $\{\bar{\bu}^{r}\}$ are bounded. On the other hand, since the objective functions of problems \eqref{eq:downlink} and \eqref{eq:uplink} are both coercive, $\{p^r\}$ and $\{q^r\}$ are also bounded. Hence, the sequence $\{\bar{\bv}^r, p^r, \bar{\bu}^r, q^r\}$ is bounded. It follows that there exists a subsequence $\{\bar{\bv}^{r_j}, p^{r_j}, \bar{\bu}^{r_j}, q^{r_j}\}$ converging to a limit point $\{\bar{\bv}^*, p^*, \bar{\bu}^*, q^*\}$. Furthermore, by further restricting to a subsequence, we can assume that the subsequence $\{\bar{\bv}^{r_j+1}, p^{r_j+1}, \bar{\bu}^{r_j+1}, q^{r_j+1}\}$ converges to some limit point $\{\bar{\bv}^{**}, p^{**}, \bar{\bu}^{**}, q^{**}\}$. According to \cite{Codreanu2007}, we have the monotonic convergence, i.e., $0\leq \sum_{k=1}^Kp_k^{r+1}\leq \sum_{k=1}^K \sigma_k^2q_k^{r+1}\leq \sum_{k=1}^K p_k^{r}\leq \sum_{k=1}^K \sigma_k^2q_k^{r}$, implying that $\sum_{k=1}^Kp_k^{**}= \sum_{k=1}^K \sigma_k^2q_k^{**}= \sum_{k=1}^K p_k^{*}= \sum_{k=1}^K \sigma_k^2q_k^{*}$.

First, we prove that $p^*$ is the optimal solution of problem \eqref{eq:downlink}  with $\bar{\bu}_k$'s and $\bar{\bv}_k$'s being fixed to $\bar{\bu}^*_k$'s and $\bar{\bv}^*_k$'s respectively, and consequently  $\sinr_k^{\textrm{D}}(p^*, \bar{\bv}^*, \bar{\bu}_k^*)= \gamma_k$, $\forall k\in\mathcal{K}$. Let $S_p=\{p~|~\sinr_k^{\textrm{D}}(p, \bar{\bv}^*, \bar{\bu}_k^*)> \gamma_k, k\in \mathcal{K}\}$. Due to the continuity of the SINR functions, we can always find an integer $J$, such that for all $j\geq J$,
$$\sinr_k^{\textrm{D}}(p, \bar{\bv}^{r_j}, \bar{\bu}_k^{r_j})\geq \gamma_k, \forall k\in \mathcal{K}, \forall p\in S_p.$$
Step 11 of the UDD algorithm [see (a) in Fig. \ref{fig:udd-flow-chart}] implies
$$\sinr_k^{\textrm{D}}(p^{r_j}, \bar{\bv}^{r_j}, \bar{\bu}_k^{r_j})\geq \gamma_k, \forall k\in \mathcal{K}$$ and $$\sum_{k=1}^K p_k^{r_j}\leq \sum_{k=1}^K p_k, \forall p\in S_p.$$
Taking limit as $j\rightarrow\infty$ yields
\begin{equation}\label{eq:SINR_limit}\sinr_k^{\textrm{D}}(p^*, \bar{\bv}^*, \bar{\bu}_k^*)\geq \gamma_k, \forall k\in\mathcal{K}\end{equation} and
\begin{equation}\label{eq:optimum_limit}\sum_{k=1}^K p_k^*\leq \sum_{k=1}^K p_k, \forall p\in S_p.
\end{equation}
Since at the optimality of problem \eqref{eq:downlink} the SINR constraint must hold with equality, \eqref{eq:SINR_limit} and \eqref{eq:optimum_limit} implies
\begin{equation}
\label{eq:downlink_SINR_eq}
\sinr_k^{\textrm{D}}(p^*, \bar{\bv}^*, \bar{\bu}_k^*)= \gamma_k, \forall k\in\mathcal{K}.
\end{equation}

Similarly, we next show that $q^{**}$ is the optimal solution of problem \eqref{eq:uplink}  with $\bar{\bu}_k$'s and $\bar{\bv}_k$'s being fixed to $\bar{\bu}^{**}_k$'s and $\bar{\bv}^*_k$'s respectively, and moreover $\sinr_k^{\textrm{U}}(q^{**}, \bar{\bu}^{**}, \bar{\bv}_k^*)= \gamma_k$, $\forall k\in\mathcal{K}$. Let $S_q =\{ q~|~\sinr_k^{\textrm{U}}(q, \bar{\bu}^{**}, \bar{\bv}_k^{*})> \gamma_k, \forall k\in\mathcal{K}\}.$
Due to the continuity of the SINR functions, we can always find $J$, for all $j\geq J$, such that
$$\sinr_k^{\textrm{U}}(q, \bar{\bu}^{r_j+1}, \bar{\bv}_k^{r_j})\geq \gamma_k, \forall k\in\mathcal{K}, \forall q\in S_q.$$
Due to Step 7 of the UDD algorithm [see (c) in Fig. \ref{fig:udd-flow-chart}], we have
$$\sinr_k^{\textrm{U}}(q^{r_j+1}, \bar{\bu}^{r_j+1}, \bar{\bv}_k^{r_j})\geq \gamma_k, \forall k\in\mathcal{K}.$$
and
$$\sum_{k=1}^K \sigma_k^2q_k^{r_j+1}\leq \sum_{k=1}^K \sigma_k^2q, \forall q\in S_q.$$
Hence, we infer that
\begin{equation}\label{eq:optimum_dual_uplink}
\sum_{k=1}^K \sigma_k^2q_k^{**}\leq \sum_{k=1}^K \sigma_k^2 q, \forall q\in S_q
\end{equation}
and
\begin{equation}\label{eq:uplink_SINR_eq}
\sinr_k^{\textrm{U}}(q^{**}, \bar{\bu}^{**}, \bar{\bv}_k^{*})= \gamma_k, \forall k\in\mathcal{K}.
\end{equation}

Now we show $\bar{\bu}^*=\bar{\bu}^{**}$ after doing an appropriate phase rotation. Due to the feasible initialization, we have for any $r_j$ that [see (b) in Fig. \ref{fig:udd-flow-chart}]
$$\sinr_k^{\textrm{D}}(p^{r_j}, \bar{\bv}^{r_j}, \bar{\bu}_k^{r_j+1})\geq \gamma_k.$$
Taking limit as $j\rightarrow\infty$ yields
\begin{align}
\sinr_k^{\textrm{D}}(p^{*}, \bar{\bv}^{*}, \bar{\bu}_k^{**})\geq\gamma_k=\sinr_k^{\textrm{D}}(p^{*}, \bar{\bv}^{*}, \bar{\bu}_k^{*}),\forall k\in\mathcal{K}\label{eq:InequalityIntermediate}
\end{align}
where the equality is due to \eqref{eq:downlink_SINR_eq}. In the following, we show that the inequality in \eqref{eq:InequalityIntermediate} actually achieves equality. Define $\hat{\gamma}_k\triangleq\sinr_k^{\textrm{D}}(p^{*}, \bar{\bv}^{*}, \bar{\bu}_k^{**})$ and assume for contrary that there exists one $\hat{\gamma}_k$ that is strictly greater than $\gamma_k$. By the uplink-downlink duality theory, there exists $\{\hat{q}_k\}$ such that $\sinr_k^{\textrm{U}}(\hat{q}, \bar{\bu}^{**}, \bar{\bv}_k^{*})=\hat{\gamma}_k$ for all $k$ and $\sum_{k=1}^K \sigma_k^2\hat{q}_k=\sum_{k=1}^Kp_k^*$. Since there exists one $k$ for which $\hat{\gamma}_k>\gamma_k$, the total power $\sum_{k=1}^K \sigma_k^2\hat{q}_k$ can be further decreased by reducing $\hat{q}_k$. Hence, \eqref{eq:optimum_dual_uplink} and \eqref{eq:uplink_SINR_eq} imply $\sum_{k=1}^K \sigma_k^2q_k^{**}<\sum_{k=1}^K \sigma_k^2\hat{q}_k=\sum_{k=1}^K p_k^*$. This yields a contradiction due to the fact that $\sum_k p_k^*=\sum_k \sigma_k^2q_k^{**}$.
Hence, we have
\begin{equation}\label{eq:key_down_eq}
\sinr_k^{\textrm{D}}(p^{*}, \bar{\bv}^{*}, \bar{\bu}_k^{**})=\sinr_k^{\textrm{D}}(p^{*}, \bar{\bv}^{*}, \bar{\bu}_k^{*}),~ \forall k\in \mathcal{K}.
\end{equation}
Note that Steps 5-6 of the UDD algorithm [see (b) in Fig. \ref{fig:udd-flow-chart}] imply
\begin{equation*}%\label{eq:normalized_downlink_transmit_beamforer}
\bar{\bu}_k^{**}=\bar{\alpha}_k^*\left(\sum_{j\neq k}\mathbf{H}_k\bv_j^*(\bv_j^*)^H\mathbf{H}_k^H +\sigma_k^2\mathbf{I}\right)^{-1}\mathbf{H}_k\bv_k^*, \forall k\in \mathcal{K},
\end{equation*}
where $\bar{\alpha}_k^*$ is normalization factor. From the above equation, we infer that $\bar{\bu}_k^{**}$  maximizes $\sinr_k^{\textrm{D}}(p^{*}, \bar{\bv}^{*}, \bar{\bu}_k)$ with respect to $\bar{\bu}_k$. It follows from \eqref{eq:key_down_eq} that $\bar{\bu}_k^{*}$ also  maximizes $\sinr_k^{\textrm{D}}(p^{*}, \bar{\bv}^{*}, \bar{\bu}_k)$. Thus we have $\vartheta\bar{\bu}_k^{**}=\bar{\bu}_k^*, \forall k\in\mathcal{K}$ for some complex valued scalar $\vartheta$ with  $|\vartheta|=1$.

Next, we show $q^*=q^{**}$ and further $\sinr_k^{\textrm{U}}(q^*, \bar{\bu}^{*}, \bar{\bv}_k^{*})= \gamma_k$, $\forall k\in\mathcal{K}$. Step 9 of the UDD algorithm [see (d) in Fig. \ref{fig:udd-flow-chart}] implies $$\sinr_k^{\textrm{U}}(q^{r_j}, \bar{\bu}^{r_j}, \bar{\bv}_k^{r_j})\geq \gamma_k, \forall k\in\mathcal{K}.$$
Taking limit as $j\rightarrow\infty$, we have
$$
\sinr_k^{\textrm{U}}(q^{*}, \bar{\bu}^{*}, \bar{\bv}_k^{*})\geq \gamma_k, \forall k\in\mathcal{K}.
$$
and thus
\begin{equation}\label{eq:dual_uplink_sinr_inequality}
\sinr_k^{\textrm{U}}(q^{*}, \bar{\bu}^{**}, \bar{\bv}_k^{*})\geq \gamma_k, \forall k\in\mathcal{K}.
\end{equation}
by noting $\vartheta\bm{\baru}_k^{**}=\bm{\baru}_k^*$. Combining \eqref{eq:dual_uplink_sinr_inequality}, \eqref{eq:optimum_dual_uplink}, and $\sum_{k=1}^K \sigma_k^2q^*_k=\sum_{k=1}^K\sigma_k^2 q^{**}_k$, we infer that both $q_k^*$'s and $q_k^{**}$'s are the optimal solutions to problem \eqref{eq:uplink} with there $\bar{\bu}_k$'s and $\bar{\bv}_k$'s replaced by $\bar{\bu}_k^{**}$'s and $\bar{\bv}_k^*$'s respectively. Since problem \eqref{eq:uplink} has a unique solution\cite{Song2007}, we conclude $q_k^*=q_k^{**}$, $\forall k\in\mathcal{K}$, and
\begin{equation}\label{eq:uplink_SINR_eq2}
\sinr_k^{\textrm{U}}(q^*, \bar{\bu}^{*}, \bar{\bv}_k^{*})= \gamma_k, \forall k\in\mathcal{K}.
\end{equation}

Now we are ready to end up the proof. By noting $\vartheta\bar{\bu}_k^{**}=\bar{\bu}_k^*$ and $q^{**}=q^*$, Step 5-8 of the UDD algorithm implies
\begin{equation}\label{eq:downlink_transmit_beamforer}
\bm{u}_k^{*}=\alpha_k^*\left(\sum_{j\neq k}\mathbf{H}_k\bv_j^*(\bv_j^*)^H\mathbf{H}_k^H +\sigma_k^2\mathbf{I}\right)^{-1}\mathbf{H}_k\bv_k^*, \forall k\in \mathcal{K}.
\end{equation}
where $\alpha_k^*$ is a normalization factor. On the other hand, Step 9-12 of the UDD algorithm implies
\begin{equation}\label{eq:uplink_transmit_beamforer}\bm{v}_k^*=\beta_k^*\left(\mathbf{I}+\sum_{j\neq k}\mathbf{H}_j^H\bu_j^*(\bu_j^*)^H\mathbf{H}_j\right)^{-1}\mathbf{H}_k^H\bu_k^*, \forall k\in\mathcal{K}.
\end{equation}
where $\beta_k^*$ is a normalization factor.

Eqs. \eqref{eq:downlink_SINR_eq}, \eqref{eq:uplink_SINR_eq2}, \eqref{eq:downlink_transmit_beamforer}, and \eqref{eq:uplink_transmit_beamforer} can be equivalently written as
\begin{align}
&\frac{1}{\gamma_k}|(\bu_k^*)^H\mathbf{H}_k\bm{v}_k^*|^2-\sum_{j\neq k}\|(\bu_k^*)^H\mathbf{H}_k\bm{v}_j^*\|^2=\sigma_k^2\|\bu_k^*\|^2\label{eq:UDD_con1}\\
&\frac{1}{\gamma_k}|(\bv_k^*)^H\mathbf{H}_k^H\bm{u}_k^*|^2-\sum_{j\neq k}\|(\bv_k^*)^H\mathbf{H}_j^H\bm{u}_j^*\|^2=\|\bv_k^*\|^2\label{eq:UDD_con2}\\
&\bu_k^*=\alpha_k^*(\mathbf{C}_k^*)^{-1}\mathbf{H}_k\bv_k^*\label{eq:UDD_con3}\\
&\bv_k^*=\beta_k^*(\bD_k^*)^{-1}\mathbf{H}_k^H\bu_k^*\label{eq:UDD_con4}
\end{align}
where $\mathbf{C}_k^*=\sum_{j\neq k}\mathbf{H}_k\bv_j^*(\bv_j^*)^H\mathbf{H}_k^H +\sigma_k^2\mathbf{I}$, $\bD_k^*=\mathbf{I}+\sum_{j\neq k}\mathbf{H}_j^H\bu_j^*(\bu_j^*)^H\mathbf{H}_j$.% $\beta_k=\frac{\sqrt{q_k}}{\|\mathbf{C}_k^{-1}\mathbf{H}_k\bv_k\|}$, $\tau_k=\frac{\sqrt{p_k}}{\|\bD_k^{-1}\mathbf{H}_k^H\bu_k\|}$.

Combining Lemma \ref{lem:lemma4-prop1}, \eqref{eq:UDD_con1}, and \eqref{eq:UDD_con3}, we obtain
\begin{equation}\label{eq:UDD_KKT1}
\left(\mathbf{C}_k^*-\frac{1}{\gamma_k}\mathbf{H}_k\bv_k^*(\bv_k^*)^H\mathbf{H}_k^H\right)\bu_k^*=0
\end{equation}
Similarly, Lemma \ref{lem:lemma5-prop2}, \eqref{eq:UDD_con2}, and \eqref{eq:UDD_con4} imply
\begin{equation}\label{eq:UDD_KKT2}
\left(\mathbf{D}_k^*-\frac{1}{\gamma_k}\mathbf{H}_k^H\bu_k^*(\bu_k^*)^H\mathbf{H}_k\right)\bv_k^*=0
\end{equation}
It can be readily seen that, \eqref{eq:UDD_KKT1}, \eqref{eq:UDD_KKT2}, and \eqref{eq:UDD_con1} implies the KKT condition \eqref{eqKKT} with $\lambda_k=1$ for all $k$. Thus the proof is completed.
\end{proof}

\begin{remark}\label{rem:feasibility}
Although both algorithms require feasible initialization, it is easier for the MMSE-DUAL algorithm than the UDD algorithm to obtain a feasible initialization. For example, when $M\geq K$, it is guaranteed that problem (P1) with any given nonzero $\bu_k$'s is feasible (e.g., zero-forcing solution for $\bv_k$'s) and thus the MMSE-DUAL algorithm can be randomly initialized in this case. However, random initialization for the UDD algorithm in this case may fail. This is also verified with a specific example in Section VI.
\end{remark}

\section{Extension to Multiple Stream Case}
Now we consider the extension of the two algorithms to the multiple stream case. Differently from \cite{Khachan2006}, we assume that joint detection is employed at receivers. Let $\bV_k$ be the transmit beamformer for user $k$. In this case, each user's achieved rate is given by
$$R_k\triangleq\log\det\left(\bI+\bH_k\bV_k\bV_k^H\bH_k^H\bOmega_k^{-1}\right)$$
where $\bOmega_k\triangleq\sigma_k^2\bI+\sum_{j\neq k}\bH_k\bV_j\bV_j^H\bH_k^H$ is the interference plus noise covariance matrix.
We are interested in solving the following rate constrained power minimization problem
\begin{equation}
\label{eq:msc_prob1}
\begin{split}
&\min_{\bV} \sum_{k=1}^K\trace(\bV_k\bV_k^H)\\
&\st~ R_k\geq r_k, k\in\mathcal{K}
\end{split}
\end{equation}
where $r_k$ represents the rate requirement for user $k$.

It is known that the streams of user $k$ can be decoded sequentially without loss of information using MMSE receiver coupled with sequential interference cancelation (SIC) technique\cite{Varanasi1997,Liuan2011}. Indeed, it is easily verified that
\begin{equation}
R_{k}=\sum_{m=1}^{d_{k}}\log(1+\textrm{SINR}_{k,m})
\end{equation}
where
$$\textrm{SINR}_{k,m}\triangleq\frac{|\bu_{k,m}^H\bH_k\bv_{k,m}|^2}{\bu_{k,m}^H\left(\sigma_k^2\bI+\sum_{j\neq k}\sum_{i=1}^{d_j}\bH_k\bv_{j,i}\bv_{j,i}^H\bH_k^H+\sum_{i=m+1}^{d_k}\bH_k\bv_{k,i}\bv_{k,i}^H\bH_k^H\right)\bu_{k,m}},$$
$\bv_{k,m}$ is the $m$-th column of $\bV_{k}$, i.e., the transmit beamformer for stream $m$, and
$$\bu_{k,m}=\left(\sigma_k^2\bI+\sum_{j\neq k}\sum_{i=1}^{d_j}\bH_k\bv_{j,i}\bv_{j,i}^H\bH_k^H+\sum_{i=m+1}^{d_k}\bH_k\bv_{k,i}\bv_{k,i}^H\bH_k^H\right)^{-1}\bH_{k}\bv_{k,m}$$ is referred to as MMSE-SIC receiver.

Define $\gamma_{k,m}\triangleq e^{\frac{r_{k}}{d_{k}}}-1$, $m=1,2,\ldots, d_k$. With equal rate allocation $\frac{r_{k}}{d_{k}}$ across multiple streams, Liu. et. al\cite[Theorem 4]{Liuan2011} proved that the following SINR-constrained power minimization problem %{\color{red}[[explain equal rate allocation]]}
\begin{equation}
\label{eq:msc_prob2}
\begin{split}
&\min_{\bu,\bv} \sum_{k=1}^K \sum_{m=1}^{d_k} \|\bv_{k,m}\|^2\\
&\st~ \frac{|\bu_{k,m}^H\bH_k\bv_{k,m}|^2}{\bu_{k,m}^H\left(\sigma_k^2\bI+\sum_{j\neq k}\sum_{i=1}^{d_j}\bH_k\bv_{j,i}\bv_{j,i}^H\bH_k^H+\sum_{i=m+1}^{d_k}\bH_k\bv_{k,i}\bv_{k,i}^H\bH_k^H\right)\bu_{k,m}}\geq \gamma_{k,m},\\
&~~~~~~~~~~~m=1,2,\ldots,d_k, k\in \mathcal{K}.
\end{split}
\end{equation}
can achieve the same optimal solution as that of problem \eqref{eq:msc_prob1}. However, we still cannot solve problem \eqref{eq:msc_prob2} to global optimality. Fortunately, the special structure of Problem \eqref{eq:msc_prob2} allows us to apply the UDD algorithm or MMSE-DUAL algorithm. Since the UDD algorithm or MMSE-DUAL algorithm reaches a KKT point of problem \eqref{eq:msc_prob2}, a question arises: whether a KKT point of problem \eqref{eq:msc_prob2} is also a KKT point of problem \eqref{eq:msc_prob1}. In the following proposition, we show the KKT equivalence of the two problems under a mild set of conditions. The proof is relegated to appendix.
 %that, using UDD/MMSE-DUAL algorithm to solve problem \eqref{eq:msc_prob2} can even reach the KKT solution of the original problem \eqref{eq:msc_prob1}.  The proof is relegated to the Appendix.
\begin{prop}\label{prop:prop3}
Let $\{\bv_{k,m},\bu_{k,m}\}$ be a KKT point of problem \eqref{eq:msc_prob2}, and let
$$\bPsi_{k,m}\triangleq\bOmega_k+\sum_{i=m}^{d_k} \bH_k\bv_{k,i}\bv_{k,i}^H\bH_k^H.$$ Suppose $\{\bv_{k,m}\}$ satisfies $\bv_{k,1}^H\bH_k^H\bPsi_{k,m}^{-1}\bH_k\bv_{k,m}\neq 0$, $m=2,3,\ldots, d_k$, $k\in\mathcal{K}$, then
\begin{enumerate}
\item $\{\bv_{k,m}\}$  is a KKT point of problem \eqref{eq:msc_prob1}.
\item A KKT point of \eqref{eq:msc_prob1} can be obtained by solving Problem \eqref{eq:msc_prob2} using either the MMSE-DUAL or the UDD algorithm.
\end{enumerate}
\end{prop}

\begin{remark}
It can be shown that the condition $\bv_{k,1}^H\bH_k^H\bPsi_{k,m}^{-1}\bH_k\bv_{k,m}\neq 0$, $m=2,3,\ldots, d_k$, $k\in\mathcal{K}$ is equivalent to $\bu_{k,m}^H\bH_k\bv_{k,1}\neq 0$, $m=2,3,\ldots, d_k$, $k\in\mathcal{K}$; see \eqref{eq:lem-key3} in Appendix. This means that, for each user $k$, detection of all the second, third, ..., and the $d_k$-th symbols are  interfered by the first symbol, which is generally true in the multi-stream scenario.
\end{remark}

\begin{remark} In the proof of Proposition \ref{prop:prop3}, the assumption $\bv_{k,1}^H\bH_k^H\bPsi_{k,m}^{-1}\bH_k\bv_{k,m}\neq 0$, $m=2,3,\ldots, d_k$, $k\in\mathcal{K}$ is used to obtain $\lambda_{k,1}=\lambda_{k,2}=\ldots=\lambda_{k,d_k}$, $k\in\mathcal{K}$, which finally leads to Part 1). It is worthy mentioning that, if this assumption is relaxed, there indeed exists some example (shown in the end of this paper)
where $\lambda_{k,1}=\lambda_{k,2}=\ldots=\lambda_{k,d_k}$, $k\in\mathcal{K}$ may not hold and in this case the KKT point $\{\bv_{k,m}\}$ of problem \eqref{eq:msc_prob2} is not a KKT point of problem \eqref{eq:msc_prob1}.
\end{remark}

\begin{remark}
Let $\bar{K}=\sum_{k=1}^K d_k$. Similar to the single stream case, it can be shown that, for the multi-stream case, the complexity of each fixed point iteration in the MMSE-DUAL algorithm is\footnote{Note that in the multi-stream case the matrix inversion operations in the fixed point iteration can still be recursively computed. This is why we have the term $O(M^3)$ as in the single stream case.} $O(\bar{K}M^2+M^3)$, while solving the system of equations for the $\bar{K}$ dual variables in the UDD algorithm requires the complexity of $O(M\bar{K}^2+\bar{K}^3)$. Hence, in the general multi-stream case where $\sum_{k=1}^K d_k>M$, the MMSE-DUAL algorithm has a lower complexity than the UDD algorithm.
\end{remark}

\section{Simulation results}
In this section, we numerically evaluate the performance of the
MMSE-DUAL algorithm and the UDD algorithm. Note that, we here only provide the convergence performance
of the two algorithms in the single stream case but similar convergence behavior is also observed in the multiple stream case.
The QoS level and noise variance are set equally to be $\gamma$ and $\sigma^2$ across all
the users. Uncorrelated fading channel model for all channel
matrices between users and the BS are assumed. Each channel
coefficient is generated from the zero mean complex Gaussian
distribution with unit variance. We also set  $N=20$ for the number
of  fixed point iterations.

\subsection{Infeasible initialization}
In our first numerical experiment, we
study the effect of the infeasible initialization on the algorithms.
In the experiment, the BS serves two users, all equipped with two antennas,
i.e., $M=N=K=2$. We set $\gamma=10$, $\sigma^2=1$, and the channel
matrices between the BS and the two users are respectively
\begin{equation}
\mathbf{H}_1=\left[\begin{array}{cc}
 0.2097 + 0.0429i & 0.4385 + 0.1650i\\
-0.9788 + 0.1614i & 0.1543 + 0.5013i
\end{array}\right],
\end{equation}

\begin{equation}
\mathbf{H}_2=\left[\begin{array}{cc}
 -1.0800 - 0.3203i & 0.2582 + 0.1785i\\
0.1714 - 0.2729i & -0.9692 - 0.1711i
\end{array}\right].
\end{equation}

With initial transmit beamformers $$\bv_1=[0.0701 {+} 0.7443i ~~-0.3386 {+} 0.0235i]^T,$$
$$\bv_2=[-1.3709 {+} 2.0320i~~0.1491 {-} 0.0298i]^T,$$
the corresponding normalized MMSE receivers are calculated as
$$\bar{\bu}_1=[-0.7423{-}0.1885i~~-0.2951 {-} 0.5713i]^T,$$
$$\bar{\bu}_2=[0.7580 {-} 0.6429i~~-0.1084 {+} 0.0209i]^T,$$
leading to SINR values, $\sinr_1=0.1592$ and $\sinr_2=4.3871$, which are both smaller than the required SINR value $\gamma=10$. Furthermore, it can be easily checked that the linear system on $q_1$ and $q_2$ with the fixed initial beamformers
$$\frac{1}{\gamma_k}q_k|\bv_k^H\mathbf{H}_k^H\bar{\bm{u}}_k|^2-\sum_{j\neq k}q_j\|\bv_k^H\mathbf{H}_j^H\bar{\bm{u}}_j\|^2=\|\bv_k\|^2,k=1,2$$
has negative solutions $q_1=-3.5627$ and $q_2=-1.1379$. Therefore, the uplink power update cannot be done in the UDD algorithm. While the UDD algorithm fails with this initialization, the MMSE-DUAL algorithm can quickly reach a feasible point in the second iteration and then exhibits a monotonic convergence in subsequent iterations, as shown in Fig. \ref{fig:fig2}.
\begin{figure}[t]
\centering
\includegraphics[width=2.7in]{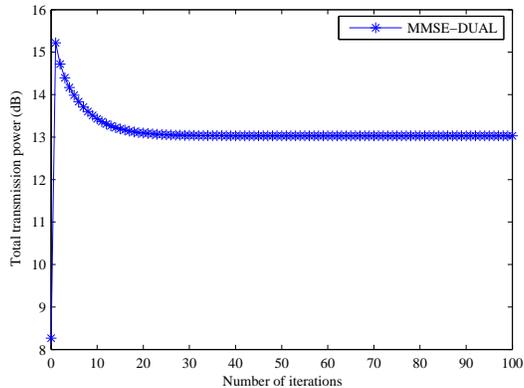}%{UDD_infeasible.eps}
\caption{Infeasible initialization of MMSE-DUAL algorithm}
\label{fig:fig2}
\end{figure}

\subsection{Convergence property}
In this set of numerical experiments, we randomly initialize the MMSE-DUAL algorithm. The UDD algorithm is initialized by a feasible point obtained by few iterations of the MMSE-DUAL algorithm. Figure \ref{fig:fig3} shows that the MMSE-DUAL algorithm and the UDD algorithm have a very similar convergence behavior.
\begin{figure}[t]
\centering
\includegraphics[width=2.7in]{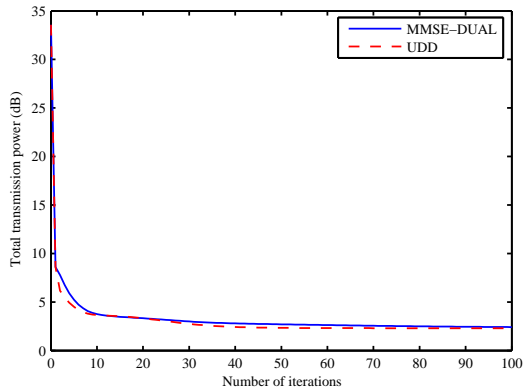}
\caption{Similar behavior of MMSE-DUAL and UDD algorithm: $K=4$, $M=7$, $N=3$, $\gamma=10$, $\sigma^2=1$.}
\label{fig:fig3}
\end{figure}

Figures \ref{fig:fig4} and \ref{fig:fig5} show that the two algorithms can almost always converge to a same objective
function value which may be global optimum regardless of initial points (different initial points are denoted
by circles). However, in an extremely rare case, local convergence for the MMSE-DUAL algorithm (also for the UDD algorithm) was observed in Fig. \ref{fig:fig6} where two different initial points resulted in two different objective values upon convergence.
\begin{figure}[t]
\centering
\includegraphics[width=3.in]{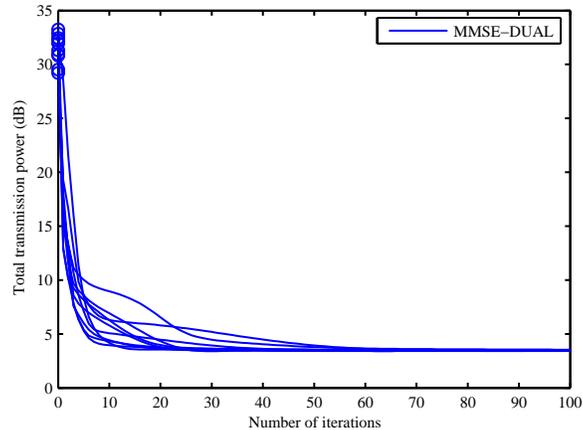}%{MMSE_DUAL_343_gamma10_noise1.eps}
\caption{The MMSE-DUAL algorithm often converge to a same objective value: $K=3$, $M=4$, $N=3$, $\gamma=10$, $\sigma^2=1$.}
\label{fig:fig4}
\end{figure}

\begin{figure}[t]
\centering
\includegraphics[width=3.in]{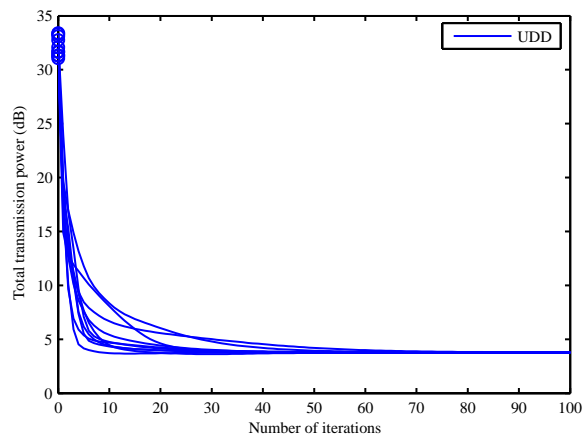}%{UDD_343_gamma10_noise1.eps}
\caption{The UDD algorithm often converge to a same objective value: $K=3$, $M=4$, $N=3$, $\gamma=10$, $\sigma^2=1$.}
\label{fig:fig5}
\end{figure}

\begin{figure}[t]
\centering
\includegraphics[width=3.in]{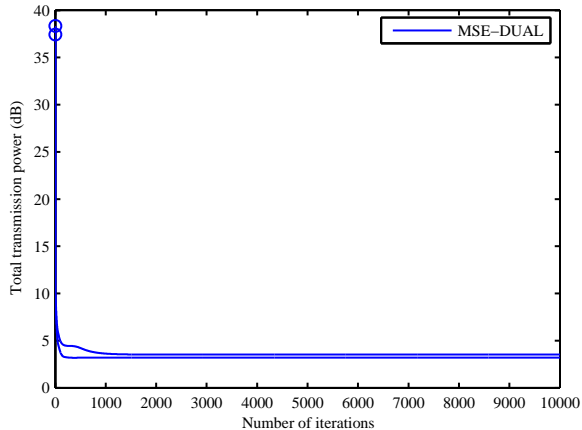}%{MMSE_DUAL_local_converge.eps}
\caption{The MMSE-DUAL algorithm (also the UDD algorithm) may converge to different local solutions depending on the initialization: $K=6$, $M=10$, $N=2$, $\gamma=100$, $\sigma^2=0.1$.}
\label{fig:fig6}
\end{figure}

\section{Conclusion}
In this paper, we have considered the SINR--constrained power minimization problem for MU-MIMO system.
Based on the KKT analysis of the power minimization problem, we propose the MMSE-DUAL algorithm. Although the latter algorithm is in essence the MMSE-SOCP algorithm in \cite{Wong2005}, it connects the MMSE-SOCP algorithm and the UDD algorithm. It is shown that the UDD algorithm\cite{Codreanu2007} also works towards
the KKT condition of the QCPM problem as the MMSE-DUAL/the MMSE-SOCP algorithm. Furthermore, we theoretically
prove that all three algorithms can monotonically converge to a KKT point of the
QCPM problem. Our numerical experiments show that the three algorithms almost always converge to a
same value which may be the optimal value, but local convergence of these algorithms is also observed.

\section{Appendix}
\subsection{The Proof Of Proposition 3}
%First we give a lemma.
%\begin{lemma}
%Problem \eqref{eq:msc_prob2} is equivalent to the following problem
%\begin{equation}
%\label{eq:msc_prob3}
%\begin{split}
%&\min_{\bV} \sum_{k=1}^K \sum_{m=1}^{d_k} \|\bv_{k,m}\|^2\\
%&\st~ \log\det\left(\bI+\bH_k\bv_{k,m}\bv_{k,m}^H\bH_k^H\left(\bOmega_k+\sum_{i=m+1}^{d_k} \bH_k\bv_{k,i}\bv_{k,i}^H\bH_k^H\right)^{-1}\right)\geq \frac{r_k}{d_k},\\
%&~~~~~~ m=1,2,\ldots, d_k, k\in\mathcal{K}
%\end{split}
%\end{equation}
%in the sense that Problem \eqref{eq:msc_prob2} and \eqref{eq:msc_prob3} share the same KKT points w.r.t. $\{\bv_{k,m}\}$.
%\end{lemma}
%{\color{red}[[The proof is again a bit too complex for people to follow. Need to clearly state which step is proving what. I will suggest you write down the kkt for problem \eqref{eq:msc_prob1} first. Then you can sort of see at each step you want to reduce which condition to which.]]}
\begin{proof}
The proof of the second part of Proposition 3 is trivial once the first part is proven. Hence, we here focus on proving the first part of Proposition 3. Let $\lambda_{k,m}$ be the Lagrange multiplier associated with the constraint of problem \eqref{eq:msc_prob2} indexed by $(k,m)$. The KKT condition of Problem \eqref{eq:msc_prob2} is given by
\begin{subequations}
\begin{align}
&\left(\bPsi_{k,m+1}-\frac{1}{\gamma_{k,m}}\bH_k\bv_{k,m}\bv_{k,m}^H\bH_k^H\right)\bu_{k,m}=0,\label{eq:lem-KKT1}\\
&\left(\bI+\sum_{j\neq k}\sum_{i=1}^{d_j}\lambda_{j,i}\bH_j^H\bu_{j,i}\bu_{j,i}^H\bH_j+\sum_{i=1}^{m-1}\lambda_{k,i}\bH_k^H\bu_{k,i}\bu_{k,i}^H\bH_k-\frac{\lambda_{k,m}}{\gamma_{k,m}}\bH_k^H\bu_{k,m}\bu_{k,m}^H\bH_k\right)\bv_{k,m}=0,\label{eq:lem-KKT2}\\
&\frac{|\bu_{k,m}^H\bH_k\bv_{k,m}|^2}{\bu_{k,m}^H\bPsi_{k,m+1}\bu_{k,m}}=\gamma_{k,m},\label{eq:lem-KKT3}\\
&\lambda_{k,m}\geq 0, \forall m, k.
\end{align}
\end{subequations}
In the following, we are going to show that given any tuple $\{\bv_{k,m}, \bu_{k,m}, {\lambda}_{k,m}\}$ satisfying the above system, $\{\bv_{k,m},{\lambda}_{k,m}\}$ must also satisfy the KKT condition for problem \eqref{eq:msc_prob1}. The proof is divided into two steps.

In the first step, we prove that \eqref{eq:lem-KKT2} can be equivalently rewritten as the first--order optimality condition of problem \eqref{eq:msc_prob1} by using \eqref{eq:lem-KKT2} and \eqref{eq:lem-KKT3}, and showing $\lambda_{k,1}=\lambda_{k,2}=\ldots=\lambda_{k,d_k}$, $\forall k$. From \eqref{eq:lem-KKT1}, we infer that $\bu_{k,m}$ must be in the form of MMSE receiver
\begin{equation}
\bu_{k,m}=\beta_{k,m}\bPsi_{k,m+1}^{-1}\bH_k\bv_{k,m}\label{lemma-mmse-receiver}
\end{equation}
where $\beta_{k,m}$ is an arbitrary nonzero scalar.
By taking $\beta_{k,m}=\frac{1}{\sqrt{1+\gamma_{k,m}}}$ and using \eqref{lemma-mmse-receiver}, we have
\begin{equation}\label{eq:lem-key1}
\bH_k^H\bu_{k,m}\bu_{k,m}^H\bH_k=\frac{1}{1+\gamma_{k,m}}\bH_k^H\bPsi_{k,m+1}^{-1}\bH_k\bv_{k,m}\bv_{k,m}^H\bH_k^H\bPsi_{k,m+1}^{-1}\bH_k.
\end{equation}
Moreover, from \eqref{eq:lem-KKT3} and \eqref{lemma-mmse-receiver}, we have
\begin{equation}\label{eq:rate_eq_appendix}
\gamma_{k,m}=\bv_{k,m}^H\bH_k^H\bPsi_{k,m+1}^{-1}\bH_k\bv_{k,m}.
\end{equation}
It follows that
\begin{equation}\label{eq:lem-key2}
\begin{split}
&\bH_k^H\bu_{k,m}\bu_{k,m}^H\bH_k\\
=&\frac{1}{1+\bv_{k,m}^H\bH_k^H\bPsi_{k,m+1}^{-1}\bH_k\bv_{k,m}}\bH_k^H\bPsi_{k,m+1}^{-1}\bH_k\bv_{k,m}\bv_{k,m}^H\bH_k^H\bPsi_{k,m+1}^{-1}\bH_k\\
=&\bH_k^H\left(\bPsi_{k,m+1}^{-1}-\left(\bPsi_{k,m+1}+\bH_k\bv_{k,m}\bv_{k,m}^H\bH_k^H\right)^{-1}\right)\bH_k\\
=&\bH_k^H\left(\bPsi_{k,m+1}^{-1}-\bPsi_{k,m}^{-1}\right)\bH_k
\end{split}
\end{equation}
where in the second equality we have used Woodbury identity\cite{Mtx_book}.
Similarly, we have
\begin{equation}
\label{eq:lem-key3}
\begin{split}
&\bH_k^H\bu_{k,m}\bu_{k,m}^H\bH_k\bv_{k,m}\\
=&\frac{\gamma_{k,m}}{1+\gamma_{k,m}}\bH_k^H\bPsi_{k,m+1}^{-1}\bH_k\bv_{k,m}\\
=&\gamma_{k,m}\left(1-\frac{\gamma_{k,m}}{1+\gamma_{k,m}}\right)\bH_k^H\bPsi_{k,m+1}^{-1}\bH_k\bv_{k,m}\\
=&\gamma_{k,m}\left(1-\frac{\bv_{k,m}^H\bH_k^H\bPsi_{k,m+1}^{-1}\bH_k\bv_{k,m}}{1+\gamma_{k,m}}\right)\bH_k^H\bPsi_{k,m+1}^{-1}\bH_k\bv_{k,m}\\
=&\gamma_{k,m}\left(\bH_k^H\bPsi_{k,m+1}^{-1}\bH_k\bv_{k,m}-\frac{\bH_k^H\bPsi_{k,m+1}^{-1}\bH_k\bv_{k,m}\bv_{k,m}^H\bH_k^H\bPsi_{k,m+1}^{-1}\bH_k\bv_{k,m}}{1+\gamma_{k,m}}\right)\\
=&\gamma_{k,m}\bH_k^H\left(\bPsi_{k,m+1}^{-1}-\frac{\bPsi_{k,m+1}^{-1}\bH_k\bv_{k,m}\bv_{k,m}^H\bH_k^H\bPsi_{k,m+1}^{-1}}{1+\bv_{k,m}^H\bH_k^H\bPsi_{k,m+1}^{-1}\bH_k\bv_{k,m}}\right)\bH_k\bv_{k,m}\\
=&\gamma_{k,m}\bH_k^H\bPsi_{k,m}^{-1}\bH_k\bv_{k,m}
\end{split}
\end{equation}
where we have used the relation $\gamma_{k,m}=\bv_{k,m}^H\bH_k^H\bPsi_{k,m+1}^{-1}\bH_k\bv_{k,m}$ in the first and third equality, and the Woodbury identity in the last equality.

By using \eqref{eq:lem-key1}, we have
\begin{equation}\label{eq:prop-key1}
\sum_{j\neq k}\sum_{i=1}^{d_j}\lambda_{j,i}\bH_j^H\bu_{j,i}\bu_{j,i}^H\bH_j=\sum_{j\neq k}\sum_{i=1}^{d_j}\lambda_{j,i}\frac{1}{1+\gamma_{j,i}}\bH_j^H\bPsi_{j,i+1}^{-1}\bH_j\bv_{j,i}\bv_{j,i}^H\bH_j^H\bPsi_{j,i+1}^{-1}\bH_j,
\end{equation}
Using \eqref{eq:lem-key2}, we get
\begin{equation}\label{eq:prop-key2}
\sum_{i=1}^{m-1}\lambda_{k,i}\bH_k^H\bu_{k,i}\bu_{k,i}^H\bH_k=\sum_{i=1}^{m-1}\lambda_{k,i}\bH_k^H\left(\bPsi_{k,i+1}^{-1}-\bPsi_{k,i}^{-1}\right)\bH_k
\end{equation}
and using \eqref{eq:lem-key3}, we obtain
\begin{equation}\label{eq:prop-key3}
\frac{\lambda_{k,m}}{\gamma_{k,m}}\bH_k^H\bu_{k,m}\bu_{k,m}^H\bH_k\bv_{k,m}=\lambda_{k,m}\bH_k^H\bPsi_{k,m}^{-1}\bH_k\bv_{k,m}
\end{equation}
By defining $$\bUpsilon_{k,i}\triangleq\frac{1}{1+\gamma_{k,i}}\bH_k^H\bPsi_{k,i+1}^{-1}\bH_k\bv_{k,i}\bv_{k,i}^H\bH_k^H\bPsi_{k,i+1}^{-1}\bH_k,$$
and substituting \eqref{eq:prop-key1}, \eqref{eq:prop-key2}, and \eqref{eq:prop-key3} into \eqref{eq:lem-KKT2}, we have
\begin{equation}\label{eq:prop-KKT-1storder}
\left(\bI+\sum_{j\neq k}\sum_{i=1}^{d_j}\lambda_{j,i}\bUpsilon_{j,i}+\sum_{i=1}^{m-1}\lambda_{k,i}\bH_k^H\left(\bPsi_{k,i+1}^{-1}-\bPsi_{k,i}^{-1}\right)\bH_k-\lambda_{k,m}\bH_k^H\bPsi_{k,m}^{-1}\bH_k\right)\bv_{k,m}=0
\end{equation}

By noting that
\begin{equation}
\begin{split}
&\sum_{i=1}^{m-1}\lambda_{k,i}\bH_k^H\left(\bPsi_{k,i+1}^{-1}-\bPsi_{k,i}^{-1}\right)\bH_k-\lambda_{k,m}\bH_k^H\bPsi_{k,m}^{-1}\bH_k\\
=&-\lambda_{k,1}\bH_k^H\bPsi_{k,1}^{-1}\bH_k+\sum_{i=1}^{m-1}(\lambda_{k,i}-\lambda_{k,i+1})\bH_k^H\bPsi_{k,i+1}^{-1}\bH_k,
\end{split}
\end{equation}
we rewrite \eqref{eq:prop-KKT-1storder} as
\begin{equation}\label{eq:prop-KKT-1storder2}
\left(\bI-\lambda_{k,1}\bH_k^H\bPsi_{k,1}^{-1}\bH_k+\sum_{j\neq k}\sum_{i=1}^{d_j}\lambda_{j,i}\bUpsilon_{j,i}+\sum_{i=1}^{m-1}(\lambda_{k,i}-\lambda_{k,i+1})\bH_k^H\bPsi_{k,i+1}^{-1}\bH_k\right)\bv_{k,m}=0.
\end{equation}
Considering \eqref{eq:prop-KKT-1storder2} with $m=1$ and $m=2$, we have
\begin{align}
&\left(\bI-\lambda_{k,1}\bH_k^H\bPsi_{k,1}^{-1}\bH_k+\sum_{j\neq k}\sum_{i=1}^{d_j}\lambda_{j,i}\bUpsilon_{j,i}\right)\bv_{k,1}=0\label{eq:prop701}\\
&\left(\bI-\lambda_{k,1}\bH_k^H\bPsi_{k,1}^{-1}\bH_k+\sum_{j\neq k}\sum_{i=1}^{d_j}\lambda_{j,i}\bUpsilon_{j,i}+(\lambda_{k,1}-\lambda_{k,2})\bH_k^H\bPsi_{k,2}^{-1}\bH_k\right)\bv_{k,2}=0\label{eq:prop702}.
\end{align}
Left-multiplying $\bv_{k,1}^H$ on both sides of \eqref{eq:prop702} and using \eqref{eq:prop701} yeild
$$(\lambda_{k,1}-\lambda_{k,2})\bv_{k,1}^H\bH_k^H\bPsi_{k,2}^{-1}\bH_k\bv_{k,2}=0$$
implying $\lambda_{k,1}=\lambda_{k,2}$ due to $\bv_{k,1}^H\bH_k^H\bPsi_{k,2}^{-1}\bH_k\bv_{k,2}\neq 0$ by assumption. Recursively, we infer from \eqref{eq:prop-KKT-1storder2} that $\lambda_{k,i}=\lambda_{k,i+1}$, by using the condition $\bv_{k,1}^H\bH_k^H\bPsi_{k,i}^{-1}\bH_k\bv_{k,i}\neq 0$, $i=2,3,\ldots,d_k$. Hence, by letting $\lambda_k=\lambda_{k,1}$ for all $k$, and using $\bV_k=[\bv_{k,1}~\bv_{k,2}~\ldots~\bv_{k,d_k}]$, we can compactly write \eqref{eq:prop-KKT-1storder2} as
\begin{equation}\label{eq:msc-pm-KKT-optimality}
\left(\bI-\lambda_{k}\bH_k^H\bPsi_{k,1}^{-1}\bH_k+\sum_{j\neq k}\lambda_{j}\sum_{i=1}^{d_j}\bUpsilon_{j,i}\right)\bV_{k}=0.
\end{equation}

Note that we have
\begin{equation}\label{eq:KKT-47-1}
\begin{split}
&\nabla_{\bV_k}\log\det(\bI+\bH_k\bV_{k}\bV_{k}^H\bH_k^H\bOmega_{k}^{-1})\\
=&\nabla_{\bV_k}\left(\log\det\bPsi_{k,1}-\log\det\bOmega_k\right)\\
=&\nabla_{\bV_k}\log\det\bPsi_{k,1}=\bH_k^H\bPsi_{k,1}^{-1}\bH_k\bV_k
\end{split}
\end{equation}
and
\begin{equation}\label{eq:KKT-47-2}
\begin{split}
&\nabla_{\bV_k}\log\det(\bI+\bH_j\bV_{j}\bV_{j}^H\bH_j^H\bOmega_{j}^{-1})\\
=&\sum_{i=1}^{d_j}\nabla_{\bV_k}\log\det(\bI+\bH_j\bv_{j,i}\bv_{j,i}^H\bH_j^H\bPsi_{j,i+1}^{-1})=-\sum_{i=1}^{d_j}\bUpsilon_{j,i}\bV_k.
\end{split}
\end{equation}
Substituting \eqref{eq:KKT-47-1} and \eqref{eq:KKT-47-2} into \eqref{eq:msc-pm-KKT-optimality} yields
\begin{equation}\label{eq:KKT-47}
\bV_k-\lambda_k\nabla_{\bV_k}\log\det(\bI+\bH_k\bV_{k}\bV_{k}^H\bH_k^H\bOmega_{k}^{-1})-\sum_{j\neq k}\lambda_j\nabla_{\bV_k}\log\det(\bI+\bH_j\bV_{j}\bV_{j}^H\bH_j^H\bOmega_{j}^{-1})=0
\end{equation}
which is the first--order optimality condition of problem \eqref{eq:msc_prob1}.

In the second step, we prove by using \eqref{eq:rate_eq_appendix} that $\{\bv_{k,m}\}$ satisfies the constraints of problem \eqref{eq:msc_prob1} with equality. Since $\gamma_{k,m}= 2^{\frac{r_{k}}{d_{k}}}-1$, $m=1,2,\ldots, d_k$, $\forall k$, we have
\begin{equation}\label{eq:prop-rate}
\begin{split}
&\log\det(\bI+\bH_k\bv_{k,m}\bv_{k,m}^H\bH_k\bPsi_{k,m+1}^{-1})\\
=&\log\det(1+\bv_{k,m}^H\bH_k\bPsi_{k,m+1}^{-1}\bH_k\bv_{k,m})\\
=&\log\det(1+\gamma_{k,m})=\frac{r_k}{d_k}.
\end{split}
\end{equation}
where we use the identity $\det(\bI+\bA\bB)=\det(\bI+\bB\bA)$\cite{Mtx_book} in the first equality and \eqref{eq:rate_eq_appendix} in the second equality. By summing \eqref{eq:prop-rate} over $m=1,2,\ldots, d_k$, we obtain
\begin{equation}\label{eq:prop-KKT-feasibility}
\sum_{m=1}^{d_k}\log\det(\bI+\bH_k\bv_{k,m}\bv_{k,m}^H\bH_k\bPsi_{k,m+1}^{-1})=\log\det(\bI+\bH_k\bV_{k}\bV_{k}^H\bH_k\bOmega_{k}^{-1})=r_k
\end{equation}

Combining \eqref{eq:KKT-47}, \eqref{eq:prop-KKT-feasibility}, and together with $\lambda_k\geq 0$ $\forall k$, we concludes that $\{\bv_{k,m}\}$ and $\{\lambda_k\}$ satisfy the KKT condition of problem \eqref{eq:msc_prob1}.
\end{proof}

%\bibliographystyle{IEEEbib}

%{
%\bibliography{ref}
%%\bibliography{biblio}
%}
\newpage
\section{A Counter Example}
In the proof of Proposition \ref{prop:prop3}, the assumption $\bv_{k,1}^H\bH_k^H\bPsi_{k,m}^{-1}\bH_k\bv_{k,m}\neq 0$, $m=2,3,\ldots, d_k$, $k\in\mathcal{K}$ is used to obtain $\lambda_{k,1}=\lambda_{k,2}=\ldots=\lambda_{k,d_k}$, $k\in\mathcal{K}$, which finally leads to Part 1) of Proposition \ref{prop:prop3}. Intuitively, it may hold that $\lambda_{k,1}=\lambda_{k,2}=\ldots=\lambda_{k,d_k}$ for the case when $\gamma_{k,1}=\gamma_{k,2}=\ldots=\gamma_{k,d_k}$, even if the assumption does not hold. However, it is worthy mentioning that, if the assumption is relaxed, there indeed exists some counter example as shown below, where, $\lambda_{k,1}=\lambda_{k,2}=\ldots=\lambda_{k,d_k}$, $k\in\mathcal{K}$ may not hold when $\gamma_{k,1}=\gamma_{k,2}=\ldots=\gamma_{k,d_k}$, $k\in\mathcal{K}$ and as a result the KKT point $\{\bv_{k,m}\}$ of problem \eqref{eq:msc_prob2} is not a KKT point of problem \eqref{eq:msc_prob1}.

%Here is a counter example to show that $\lambda_{k,i}$, $i=1,2,\ldots, d_k$, may not be equal when $\gamma_{k,i}$, $i=1,2,\ldots, d_k$, are equal.
Let us consider the special case---the point to point MIMO system with two streams. In this case, the rate-constrained power minimization problem \eqref{eq:msc_prob1} boils down to
\begin{equation}\label{eq:example}
\begin{split}
&\min_{\bV} \trace(\bV\bV^H)\\
&\st~\log\det(\bI+\bH\bV\bV^H\bH)\geq r.
\end{split}
\end{equation}
With equal rate allocation for the two streams, the optimal solution to problem \eqref{eq:example} can be found by solving the following problem
\begin{equation}\label{eq:example2}
\begin{split}
&\min_{\bu_1,\bu_2, \bv_1, \bv_2} \Vert\bv_1\Vert^2+\Vert\bv_2\Vert^2\\
&\st~\frac{\vert\bu_1^H\bH\bv_1\vert^2}{\bu_1^H(\bI+\bH\bv_2\bv_2^H\bH^H)\bu_1}\geq \gamma_1,\\
&~~~~~~\frac{\vert\bu_1^H\bH\bv_2\vert^2}{\bu_1^H\bu_1}\geq \gamma_2
\end{split}
\end{equation}
where $\gamma_1=\gamma_2=e^{\frac{r}{2}}-1$.

Let $\lambda_1$ and $\lambda_2$ be the Lagrange multiplier associated with the first and second constraint of problem \eqref{eq:example2}, respectively. In the following, we show that there may exist the case $\lambda_1\neq \lambda_2$ when $\gamma_1=\gamma_2$. The KKT system of problem \eqref{eq:example2} is given by
\begin{align}
&\left(\bI+\bH\bv_2\bv_2^H\bH^H-\frac{1}{\gamma_1}\bH\bv_1\bv_1^H\bH^H\right)\bu_1=0\label{exkkt1}\\
&\left(\bI-\frac{1}{\gamma_2}\bH\bv_2\bv_2^H\bH^H\right)\bu_2=0\label{exkkt2}\\
&\left(\bI-\frac{\lambda_1}{\gamma_1}\bH^H\bu_1\bu_1^H\bH\right)\bv_1=0\label{exkkt3}\\
&\left(\bI+\lambda_1\bH^H\bu_1\bu_1^H\bH-\frac{\lambda_2}{\gamma_2}\bH^H\bu_2\bu_2^H\bH\right)\bv_2=0\label{exkkt4}\\
&\frac{\vert\bu_1^H\bH\bv_1\vert^2}{\bu_1^H(\bI+\bH\bv_2\bv_2^H\bH^H)\bu_1}=\gamma_1\label{exkkt5}\\
&\frac{\vert\bu_2^H\bH\bv_2\vert^2}{\bu_2^H\bu_2}=\gamma_2\label{exkkt6}\\
&\lambda_1,\lambda_2\geq 0\label{exkkt7}
\end{align}

%The counter example is inspired by the assumption made in Proposition 3.
For the special case, we relax the assumption, i.e., let $\bv_1^H\bH^H\bH\bv_2=0$, which holds if $\bv_1$ and $\bv_2$ are two orthogonal eigenvectors of $\bH^H\bH$. In the following, we let $\bv_1$ and $\bv_2$ be two orthogonal eigenvectors of $\bH^H\bH$, such that $\bv_2^H\bH^H\bH\bv_2=\gamma_2$ and $\bv_1^H\bH^H\left(\bI+\bH\bv_2\bv_2^H\bH^H\right)^{-1}\bH\bv_1=\gamma_1$. Hence, we have $\bH^H\bH\bv_1=\mu_1\bv_1$ and $\bH^H\bH\bv_2=\mu_2\bv_2$ where $\mu_1$ and $\mu_2$ are the two corresponding eigenvalues of $\bH^H\bH$. Furthermore, we let
\begin{equation}\label{eq:uu1}
\bu_1=\frac{1}{\sqrt{1+\gamma_1}}\left(\bI+\bH\bv_2\bv_2^H\bH^H\right)^{-1}\bH\bv_1,
\end{equation}
\begin{equation}\label{eq:uu2}\bu_2=\frac{1}{\sqrt{1+\gamma_2}}\bH\bv_2,\end{equation}
and
\begin{equation}\label{exlambda}
\lambda_1=\frac{1+\gamma_1}{\mu_1}~ \textrm{and}~ \lambda_2=\frac{1+\gamma_2}{\mu_2}.
\end{equation}

Now we are ready to show that $(\bu_1,\bv_1,\bu_2,\bv_2,\lambda_1,\lambda_2)$ defined above satisfies the KKT system (\ref{exkkt1}-\ref{exkkt7}). Let us start with examining \eqref{exkkt1} and \eqref{exkkt2}.
First, since $\bv_1^H\bH^H\bH\bv_2=0$, we have $\left(\bI+\bH\bv_2\bv_2^H\bH^H\right)^{-1}\bH\bv_1=\bH\bv_1$. Thus, \begin{equation}\label{eq:uuu1}\bu_1=\frac{1}{\sqrt{1+\gamma_1}}\bH\bv_1,\end{equation} which is clearly orthogonal to $\bu_2$ and $\bH\bv_2$ due to $\bv_1^H\bH^H\bH\bv_2=0$. It follows that
\begin{equation}
\begin{split}
&\left(\bI+\bH\bv_2\bv_2^H\bH^H-\frac{1}{\gamma_1}\bH\bv_1\bv_1^H\bH^H\right)\bu_1\\
%=&\bu_1-\frac{1}{\gamma_1}\bH\bv_1\bv_1^H\bH^H\bu_1\\
=&\bu_1-\frac{1}{\sqrt{1+\gamma_1}}\frac{1}{\gamma_1}\bH\bv_1\bv_1^H\bH^H\left(\bI+\bH\bv_2\bv_2^H\bH^H\right)^{-1}\bH\bv_1\\
=&\bu_1-\frac{1}{\sqrt{1+\gamma_1}}\bH\bv_1=0
\end{split}
\end{equation}
where the first equality is due to $\bv_2^H\bH^H\bu_1=0$ and \eqref{eq:uu1}, the second equality is due to\\ $\bv_1^H\bH^H\left(\bI+\bH\bv_2\bv_2^H\bH^H\right)^{-1}\bH\bv_1=\gamma_1$, and the last equality is due to
\eqref{eq:uuu1}. Thus, \eqref{exkkt1} follows. Similarly, we can verify that $\bu_2$ and $\bv_2$ satisfy \eqref{exkkt2}.

Next let us check \eqref{exkkt3} and \eqref{exkkt4}.
First, due to \eqref{eq:uu1} and $\bv_1^H\bH^H\left(\bI+\bH\bv_2\bv_2^H\bH^H\right)^{-1}\bH\bv_1=\gamma_1$, we have $\bu_1^H\bH\bv_1=\frac{\gamma_1}{\sqrt{1+\gamma_1}}$. It follows that
\begin{equation}
\begin{split}
&\left(\bI-\frac{\lambda_1}{\gamma_1}\bH^H\bu_1\bu_1^H\bH\right)\bv_1\\
&=\bv_1-\frac{\lambda_1}{\sqrt{1+\gamma_1}}\bH^H\bu_1\\
&=\bv_1-\frac{\lambda_1}{1+\gamma_1}\bH^H\bH\bv_1=0
%&=\bv_1-\frac{1}{\mu_1}\bH^H\bH^H\bv_1=0
\end{split}
\end{equation}
where the second equality is due to \eqref{eq:uuu1} and the last equality follows from $\lambda_1=\frac{1+\gamma_1}{\mu_1}$ and  $\bH^H\bH\bv_1=\mu_1\bv_1$. Thus, \eqref{exkkt3} follows. Similarly we can verify \eqref{exkkt4}.

Finally, it is easy to verify that \eqref{exkkt5}, \eqref{exkkt6}, and \eqref{exkkt7} hold true for the defined $(\bu_1,\bv_1,\bu_2,\bv_2,\lambda_1,\lambda_2)$.

Now we are ready to draw the conclusion. From \eqref{exlambda}, it is readily known that, when $\bH^H\bH$ has two different eigenvalues $\mu_1$ and $\mu_2$, we have $\lambda_1\neq\lambda_2$  even if $\gamma_1=\gamma_2$.
\end{document}